\documentclass[aps,pra,reprint,amsmath,amssymb,superscriptaddress]{revtex4-2}

\usepackage{graphicx}
\usepackage{bm}
\usepackage{hyperref}

\usepackage[english]{babel}
\usepackage{eurosym}
\usepackage{color}
\usepackage{indentfirst}
\usepackage{amsmath}
\usepackage{graphicx}
\usepackage{float}
\usepackage{rotfloat}
\usepackage{rotating}
\usepackage{adjustbox}
\usepackage{multirow}
\usepackage{amssymb}
\usepackage{indentfirst}
 \usepackage{dsfont}
 \usepackage{textcomp}
 \usepackage[T1]{fontenc}
\usepackage{lmodern}
\usepackage{pifont}
\usepackage{graphicx}
\usepackage{amsthm}
\usepackage{natbib}
\usepackage{appendix}
\usepackage{bm}
\usepackage{makecell}
\usepackage{hyperref}
\usepackage{color}

\usepackage{braket}
\usepackage[compat=0.3]{yquant}
\usepackage{algorithm}
\usepackage{algorithmicx}
\usepackage{algpseudocode}
\usepackage{tikz}
\usetikzlibrary{positioning,calc,fit,backgrounds,shadows.blur}

\usepackage[most]{tcolorbox} 
\usepackage{xcolor}
\usepackage{xltxtra}
\tcbuselibrary{breakable, skins}
\newtheorem{theorem}{Theorem}
\newtheorem*{theorem*}{Theorem}
\newtheorem{lemma}{Lemma}
\newtheorem{definition}{Definition}

\begin{document}

\title{A Lyapunov Framework for Quantum Algorithm Design in Combinatorial Optimization with Approximation Ratio Guarantees
}

\author{Shengminjie Chen}
\email[Corresponding author: ]{csmj@ict.ac.cn, sunxiaoming@ict.ac.cn}
\affiliation{State Key Lab of Processors, Institute of Computing Technology, Chinese Academy of Sciences, Beijing 100190, China}
\affiliation{School of Computer Science and Technology, University of Chinese Academy of Sciences, Beijing 100049, China}

\author{Ziyang Li}
\affiliation{State Key Lab of Processors, Institute of Computing Technology, Chinese Academy of Sciences, Beijing 100190, China}
\affiliation{School of Computer Science and Technology, University of Chinese Academy of Sciences, Beijing 100049, China}

\author{Hongyi Zhou}
\affiliation{State Key Lab of Processors, Institute of Computing Technology, Chinese Academy of Sciences, Beijing 100190, China}
\affiliation{School of Computer Science and Technology, University of Chinese Academy of Sciences, Beijing 100049, China}

\author{Jialin Zhang}
\affiliation{State Key Lab of Processors, Institute of Computing Technology, Chinese Academy of Sciences, Beijing 100190, China}
\affiliation{School of Computer Science and Technology, University of Chinese Academy of Sciences, Beijing 100049, China}

\author{Wenguo Yang}
\affiliation{School of Mathematical Sciences, University of Chinese Academy of Sciences, Beijing 100049, China}

\author{Xiaoming Sun}
\email[Corresponding author: ]{csmj@ict.ac.cn, sunxiaoming@ict.ac.cn}
\affiliation{State Key Lab of Processors, Institute of Computing Technology, Chinese Academy of Sciences, Beijing 100190, China}
\affiliation{School of Computer Science and Technology, University of Chinese Academy of Sciences, Beijing 100049, China}


\begin{abstract}

In this work, we develop a framework aiming at designing quantum algorithms for combinatorial optimization problems while providing theoretical guarantees on their approximation ratios. The principal innovative aspect of our work is the construction of a time-dependent Lyapunov function that naturally induces a controlled Schrödinger evolution with a time dependent Hamiltonian for maximizing approximation ratios of algorithms. Because the approximation ratio depends on the optimal solution, which is typically elusive and difficult to ascertain a priori, the second novel component is to construct the upper bound of the optimal solution through the current quantum state. By enforcing the non-decreasing property of this Lyapunov function, we not only derive a class of quantum dynamics that can be simulated by quantum devices but also obtain rigorous bounds on the achievable approximation ratio. As a concrete demonstration, we apply our framework to Max-Cut problem, implementing it as an adaptive variational quantum algorithm based on a Hamiltonian ansatz. This algorithm avoids ansatz and graph structural assumptions and bypasses parameter training through a tunable parameter function integrated with measurement feedback.  


\end{abstract}


\maketitle


\section{Introduction}
Combinatorial optimization constitutes a significant research area in the interdisciplinary fields, with tremendous applications in many scenarios. The fundamental target of combinatorial optimization problems is to find the optimal or near-optimal solution within a discrete and finite feasible domain associated with a specific objective function. Although some combinatorial optimization problems can be solved within polynomial complexity, a considerable number of them, characterized by their NP-hard property, exhibit an exponential growth in computational resources for obtaining an exact solution. Consequently, approximation algorithms and heuristic algorithms have emerged as viable alternatives, as they are capable of good solutions within acceptable computational resources. Recently, quantum computers have demonstrated their ability to handle some problems that are intractable for classical computers. Adopting quantum computers to solve combinatorial optimization problems has become a hot issue in quantum computing. Further insights into quantum optimization algorithms can refer to Survey \citep{Abbas2023}.

The pioneering quantum algorithm for combinatorial optimization problems is the Grover search, which explores a quadratic advantage in unstructured database and can be translated into the minimization of functions \citep{Durr1996} and the brute-force \citep{Gilliam2021}. Alternatively, Quantum Adiabatic Algorithm (QAA) \citep{Farhi2000,Kadowaki1998} encodes a problem Hamiltonian $\mathbf{H}$, whose ground quantum state is intrinsically associated with the solution. The adiabatic theorem guarantees that the state $|\psi\rangle$ will be transited from the initial state to the ground state of $\mathbf{H}$, when the annealing is slowly enough. However, the energy gap diminish exponentially with the increase of the system within the glass phase \citep{Mohseni2023}. Consequently, it entails exponentially prolonged run times, when the objective is to adhere as precisely as possible to the adiabatic path. Quantum Walk is also a quantum algorithm framework for combinatorial optimization problems such as TSP \citep{Marsh2020} and N-Queen \citep{campos2023}.

On near-term intermediate-scale quantum (NISQ) era, Variational Quantum Algorithms (VQAs) \citep{Cerezo2021} have attracted significant attentions. Formally, a VQA is a parameterized quantum circuit that adopts classical computers to iteratively determine parameters such that the objective function of quantum state measured on some bases is as large(small) as possible. Prominent instances of VQA include the Variational Quantum Eigensolver (VQE) \citep{McClean2016,Peruzzo2014} and the Quantum Approximate Optimization Algorithm (QAOA) \citep{farhi2014}. In addition, numerous variants of QAOA \citep{Egger2021,Wurtz2022,Grimsley2023,Bucher2025IFQAOA,MontanezBarrera2025,Hadfield2018,Lu2023} have been developed and can refer to Survey \citep{BLEKOS2024}. Unfortunately, finding the optimal parameters in a VQA is time-consuming in classical computers, attributed to NP-hard property \citep{Bittel2021} and barren plateaus \citep{McClean2018}. Nevertheless, several heuristic schemes have been proposed to obtain good parameters \citep{Akahoshi2024,Golden2023,Heidari2024,Ragone2024,Grimsley2019,Grimsley2023}. The structure-based technique and the measurement-based technique for parameters have been explored \citep{Wang2025,Magann2021,Magann2022a,Magann2022b}. In addition, these parameters technique has also been used in QAOA to solve multi-objective optimization problem \citep{kotil2025}. 

Previous studies have predominantly explored the efficiency and empirical performance. In parallel, and analogous to classical approximation algorithms, theoretical analyses of their approximation ratios $\frac{\langle \psi(t)|\mathbf{H}_f |\psi(t) \rangle}{\langle \psi^*|\mathbf{H}_f |\psi^* \rangle}$ have also attracted considerable attention. In \citep{farhi2014}, QAOA can obtain a $0.6924$ approximation solution for Max-Cut problem under $3$-regular graph when $p=1$ round. Building upon this, \citep{Wurtz2021} proved that QAOA with $p=2$ returns a $0.7559$ approximation solution for Max-Cut problem in $3$-regular graph and further conjectured that in the worst case when $p=3$, it yields a $0.7924$ approximation solution. For large-girth $D$-regular graphs, \citep{Basso2022,Farhi2022} provided some evidences that QAOA with $p \geq 11$ potentially outperforms classical approximation algorithms, such as the SDP-based algorithm with a 0.878 approximation ratio \citep{Goemans1995}. More recently, \citep{Li2024} extended the similar technique to a special type of low-girth graph. \citep{wang2025b} proposed a dynamic algorithm with the tree-decomposition technique for calculating the expected cut fraction. In addition, \citep{Braida2022} employed the Lieb-Robinson bound to prove that the constant-time quantum annealing returns approximation solutions under bounded degree graphs. \citep{Braida2024} demonstrated that there is a $1$-local quantum annealing algorithm with the $0.7020$ approximation ratio for Max-Cut on $3$-regular graph. \citep{Wang2025} illustrated that one round light-cone VQA is $0.7926$ approximation algorithm for Max-Cut on infinite $3$-regular graph. 

However, the aforementioned theoretical analyses for the Max-Cut problem are still subject to stringent limitations. Constructing efficient iterative formulas for the expected cut fraction typically requires imposing additional constraints on graph structures such as large-girth \citep{Basso2022,Farhi2022}, $D$-regular \citep{Basso2022,farhi2014,Farhi2022,Wurtz2021}, additive product graph \citep{Li2024}, and tree-decomposition \citep{wang2025b}. For general graphs, the theoretical guarantee for Max-Cut is still a open problem. While QAOA/VQA and Quantum Walk are unified frameworks to obtain an 'approximation solution' for combinatorial optimization problems, they pose challenges to analyze the approximation ratio in general for other combinatorial optimization problems. In this work, we endeavor to offer a partial resolution to the following issue.
  
\emph{Does there exist a partially generic framework that can be employed to design quantum algorithms while analyzing the theoretical guarantee (albeit perhaps not the tightest) for combinatorial optimization problems?}

In this work, we propose a universal framework called quantum Lyapunov framework for designing quantum algorithms for combinatorial optimization problems with provable approximation ratio guarantees. The core idea of our framework is inspired by Lyapunov framework \citep{Du2022,chen2023} on classical algorithms. Intuitively, a quantum algorithm can be described by a dynamic system from an initial state $|\psi(0)\rangle$ governed, for example, by the Schr\"{o}dinger equation with Hamiltonian $\mathbf{H}(t)$. This evolution can be viewed as a trajectory in Hilbert space originating from the initial point. Similarly, a classical algorithm for optimization problems may be regarded as a trajectory: beginning from an initial solution $\mathbf{x}(0) \in \mathbb{R}^n$ or $S_0 \subseteq N$ on definition domains $\mathbb{R}^n$ (continuous optimization problems) or $2^N$ (combinatorial optimization problems) and sequentially connecting all points with a piece-wise straight line. That implies that both quantum algorithms and classical algorithms for optimization problems can also be regarded as a dynamic system with some proper discrete procedures. 

Consequently, we construct a time-dependent Lyapunov function that incorporates both the problem dependent Hamiltonian and the (unknown) optimal solution. Enforcing the non-decreasing property of this Lyapunov function over time, we naturally derive a time-dependent Hamiltonian that governs the system's evolution such that the final state has rigorously theoretical guarantees. Discreterizing the evolution time, we establish a link between the dynamic system and an iterative process. As a result, the final state converges to a superposition of 'good solutions' for the combinatorial problem. To implement this evolution on near-term quantum devices, the continuous dynamics can be digitally simulated by adaptive variational quantum circuits with parameter schedules determined by the Lyapunov condition. As a concrete demonstration, we apply it to Max-Cut problem, designing a specific variational quantum algorithm, rigorously analyzing its approximation ratio, and fitting the convergence rate, thereby showcasing a practical pathway to achieve performance guarantees beyond limitations of existing QAOA analyses and illustrating the possible acceleration compared to classical algorithms.

\section{Results}
Compared with the analysis of approximation ratios for QAOA/VQA with some special problems characteristic limitation \citep{Basso2022,farhi2014,Farhi2022,Li2024,Wurtz2021,wang2025b}, our framework is a partially generic framework that only depends on the quantum upper bound of the optimal solution. For the Max-Cut problem, our framework diverges from existing works by alleviating the restriction such as $D$-regular and large-girth, thereby rendering it applicable to general graphs. Additionally, we also incorporate feedback-control and measurements techniques to overcome the obstacle for obtaining parameters in QAOA/VQAs.

Compared with the Lyapunov control-inspired quantum algorithms \citep{Magann2022a,Magann2022b}, while our approach regarding parameters involves a shared aspect, our framework possesses capabilities of not only designing quantum algorithms but also conducting a comprehensive analysis of their theoretical guarantees. In other words, our framework effectively fills the gap of approximation ratios for designing Lyapunov control-inspired quantum algorithms. Essentially, in situations where the approximation ratio remains undefined, the Lyapunov control-inspired quantum algorithms is a quantum heuristic algorithm scheme. However, with the implementation of our framework, they are now transformed into a quantum approximation algorithm scheme, thereby enhancing their theoretical rigor and practical applicability.   

Potentially, our framework streamlines the process of designing quantum approximation algorithms by reducing it to construct the quantum upper bound of the optimal solution. This means that designing quantum approximation algorithms with good theoretical guarantees fundamentally entails to identify a proper quantum upper bound of the optimal solution, which is predicted on properties of the objective function and the corresponding constraints. 

In this work, we consider the maximization problem. More technical details are covered in Appendix, including concretely discrete procedures and proofs.

\subsection{Lyapunov Function with One Parameter}


As a warm-start, we first introduce the Lyapunov functions with one parameter. The formal statement is as follows:
\begin{equation}\label{eq:: Lyapunov function 1 main}
    E(t) = \langle \psi(t)| \mathbf{H}_f | \psi(t) \rangle - \lambda(t) \langle \psi^*| \mathbf{H}_f | \psi^* \rangle
\end{equation}
where $\mathbf{H}_f$ is the Hamiltonian of the objective function $f$ for combinatorial optimization problems and can be expressed as a polynomial of Pauli-$Z$, i.e., $\mathbf{H}_f = c_0+\sum_{j}c_j \mathbf{Z}_j+\sum_{j<k}c_{jk} \mathbf{Z}_j \mathbf{Z}_k + \sum_{j<k<\ell}c_{jk\ell} \mathbf{Z}_j \mathbf{Z}_k \mathbf{Z}_\ell+\cdots$. $|\psi ^* \rangle$ is the quantum state of the optimal solution for combinatorial optimization problems, i.e., $|\psi ^* \rangle = \arg\max_{|\psi\rangle} \langle \psi | \mathbf{H}_f | \psi \rangle$. Concretely, if we denote $OPT = \{S \in \mathcal{C}| f(S) \ge f(V), \forall V \in \mathcal{C}\}$, where $\mathcal{C}$ represents the feasible solution space of the combinatorial optimization problem, then $|\psi^* \rangle = \sum_{S \in OPT} a_S |S\rangle$. $|\psi(t) \rangle$ is the quantum state returned by a quantum algorithm at time $t$. If we can control the above equation properly, we can indirectly obtain an appropriate Hamiltonian $\mathbf{H}(t)$ such that the quantum algorithm based on the Hamiltonian $\mathbf{H}(t)$ for Schr\"{o}dinger equation possesses a provably reliable approximation guarantee. Primarily, at each time $t$, if we can ensure that the inequality $E(t) \ge E(0)$ is maintained, then there is the following theoretical guarantee: 
\begin{equation}
        \langle \psi(t)| \mathbf{H}_f | \psi(t) \rangle \ge (\lambda(t)-\lambda(0)) \langle \psi^*| \mathbf{H}_f | \psi^* \rangle + \langle \psi(0)| \mathbf{H}_f | \psi(0) \rangle
\end{equation}
Consequently, there is a sufficient condition to ensure $E(t) \ge E(0)$, i.e., $E(t)$ is non-decreasing. If $E(t)$ is differentiable, the non-decreasing Lyapunov function $E(t)$ can be transfer to $\frac{d E(t)}{d t} \ge 0$. Considering the derivative of Lyapunov function in Eq.(\ref{eq:: Lyapunov function 1 main}), we have:
\begin{equation*}
    \cfrac{d E(t)}{d t} =\langle \psi(t)| i [\mathbf{H}^{\hat{f}}(t), \mathbf{H}_f] | \psi(t) \rangle - \cfrac{d \lambda(t)}{dt} \langle \psi^*| \mathbf{H}_f | \psi^* \rangle
\end{equation*}
Potentially, $\lambda(t)$ can be constrained to be a non-decreasing function, i.e., $\frac{d \lambda(t)}{d t} \ge 0$. This is because, as the algorithm progresses, we invariably aspire to achieve the non-decreasing approximation ratio. Additionally, $\mathbf{H}(t)$ can be partitioned into two components: (1) the part that commutes with $\mathbf{H}_f$ (denoted as $\mathbf{H}^{f}(t)$) and (2) the part that does not commute with $\mathbf{H}_f$ (denoted as $\mathbf{H}^{\hat{f}}(t)$). However, in the derivative of the Lyapunov function, it is challenging to ascertain $|\psi^*\rangle$ in advance. This poses an obstacle to ensuring the non-negativity of $\frac{d E(t)}{d t}$. To deal with this issue, we consider the upper bound of $\langle \psi^*| \mathbf{H}_f | \psi^* \rangle$, which can be easily estimated prior to the design of the algorithm. Normally, the upper bound of the optimal solution is determined by the mathematical properties of the problem. Suppose that the upper bound has the following form:
\begin{equation}
    0 \le \langle \psi^*| \mathbf{H}_f | \psi^* \rangle \le \langle \psi(t)| \mathbf{Q}(t) | \psi(t) \rangle
\end{equation}
Therefore, we have a sufficient condition to ensure the monotonicity of Lyapunov function:
\begin{equation*}
    \begin{array}{c}
        \begin{aligned}
            \cfrac{d E(t)}{d t} & \ge \langle \psi(t)| i [\mathbf{H}^{\hat{f}}(t) , \mathbf{H}_f] | \psi(t) \rangle - \cfrac{d \lambda(t)}{d t} \langle \psi(t)| \mathbf{Q}(t) | \psi(t) \rangle \\
            & = 0
        \end{aligned} 
    \end{array}
\end{equation*}
According to the above equation, it induces a limitation of $\mathbf{H}^{\hat{f}}(t)$, i.e., $\langle \psi(t)| i [\mathbf{H}^{\hat{f}}(t) , \mathbf{H}_f] | \psi(t) \rangle\ge 0 $. Therefore, the approximation ratio $\lambda(t)-\lambda(0)$ has the following explicit formula:
\begin{equation}
        \lambda(T)-\lambda(0) = \int_0^T \cfrac{\langle \psi(s)| i [\mathbf{H}^{\hat{f}}(s) , \mathbf{H}_f] | \psi(s) \rangle} { \langle \psi(s)| \mathbf{Q}(s) | \psi(s) \rangle }ds
\end{equation}
Different choices for $\mathbf{H}^{\hat{f}}(t)$ will directly lead to different algorithms and their approximation ratios. Specially, in this work, we decompose the commutative component and the non-commutative component into a combination of efficiently implementable quantum gates as below:
\begin{equation}
    \mathbf{H}^f(t) = \sum_{k=1}^{n_1}\eta_k(t)\mathbf{H}^f_k, \quad \mathbf{H}^{\hat{f}}(t) = \sum_{k=1}^{n_2} \alpha_k(t) \mathbf{H}^{\hat{f}}_k
    \label{eq::Algorithm_Hamiltonian_Choice}
\end{equation}
where we restrict that $\mathbf{H}^f_k$ and $\mathbf{H}^{\hat{f}}_k$ can be implemented on quantum circuits. $\mathbf{H}^f(t)$ and $\mathbf{H}^{\hat{f}}(t)$ described above imply that they commute at different times, respectively. By employing feedback-control and denoting $O_k(t) = \langle \psi(t)| i [\mathbf{H}^{\hat{f}}_k, \mathbf{H}_f] | \psi(t) \rangle$, we have: 
\begin{equation*}
    \lambda(T)-\lambda(0) = \int_0^T \cfrac{\sum_{k=1}^{n_2}\alpha_k(s)O_k(s)} { \langle \psi(s)| \mathbf{Q}(s) | \psi(s) \rangle }ds
\end{equation*}
Our target is reduced to choose proper $\alpha_k(t)$ such that $\lambda(T)-\lambda(0)$ is as large as possible and $\sum_{k=1}^{n_2}\alpha_k(t)O_k(t) \ge 0$. Potentially, there is an intuitive and simple choice about $\alpha_k$:
\begin{equation}
    \alpha_k(t) = \beta_k(t) O_k(t), \beta_k(t) \ge 0
    \label{eq::value_of_alpha}
\end{equation}

Totally, the Hamiltonian and the approximation ratio corresponding to the quantum algorithm are as follows:
\begin{equation}
    \begin{array}{c}
        \mathbf{H}(t) = \sum_{k=1}^{n_1} \eta_k(t)\mathbf{H}^f_k + \sum_{k=1}^{n_2} \beta_k(t) O_k(t)\mathbf{H}^{\hat{f}}_k \\
        \lambda(T)-\lambda(0) = \displaystyle\int_0^T \cfrac{\sum_{k=1}^{n_2}\beta_k(t) O_k^2(t)} { \langle \psi(s)| \mathbf{Q}(s) | \psi(s) \rangle }ds
    \end{array}
    \label{eq::single_para_estimate}
\end{equation}

The continuous-time algorithm defined in Eq. (\ref{eq::single_para_estimate}) can be discretized into a corresponding parameterized ansatz. Firstly, we discretize continuous time interval $[0,T]$ into a discrete time series $t_0=0< t_1< ... < t_N=T$. Within each subinterval $t\in (t_j,t_{j+1}]$, the time-dependent Hamiltonian $\mathbf{H}(t)$ is approximated as a time-independent Hamiltonian $\mathbf{H}(t_j)$. Adopting Trotter decomposition, we then express the time evolution operator as $e^{i \mathbf{H}(t_j)(t_{j+1}-t_j)} \approx U^{\hat{f}}(j+1) U^f(j+1)$ where $U^f(j+1) = \prod_{k=1}^{n_1} e^{-i\eta_k(j+1)\mathbf{H}^f_k (t_{j+1} - t_{j})}$ and $U^{\hat{f}}(j+1) = \prod_{k=1}^{n_2} e^{-i  \alpha_k(j+1) \mathbf{H}_k^{\hat{f}} (t_{j+1} - t_{j})}$. The resulting pseudo-algorithm is outlined in Fig.\ref{fig:algorithm}. In Appendix \ref{Appendix::One-Parameter-Discretion}, we provide detailed descriptions for discretized procedure of this algorithm and analyze the increment in the approximation ratio achieved during each iteration of the discretized algorithm, which is formally characterized by the following theorem.

\begin{theorem}
In the discretized version of the algorithm defined by Eq. (\ref{eq::single_para_estimate}), if the evolution time $t_{j+1}-t_j \le \frac{\epsilon}{B+C\epsilon+\sqrt{A \epsilon}}$, then the improvement in the approximation ratio contributed by the $j+1$-th iteration can be expressed as below with the bounded discretization error $\epsilon$
    \begin{equation}
        \lambda(t_{j+1} ) - \lambda(t_j) = \cfrac{\sum_{k=1}^{n_2} \beta_k(t_j)O_k^2(t_j) (t_{j+1}-t_j)}{\langle \psi(t_j)| \mathbf{Q}(t_j) | \psi(t_j) \rangle}
    \end{equation}
    where
    \begin{equation*}
    \begin{array}{l}
        A = 2 \left\| \mathbf{H}_f \right\|
               \left\| \mathbf{H}^{\hat{f}}(t_j) \right\|
               \left\| \mathbf{H}(t_j) \right\| \\
        \begin{aligned}
            B = 2 \left\| \mathbf{H}_f \right\| 
             \bigg(
               & \sum_{k=1}^{n_1} \left\| \mathbf{H}^f_{k}(t_j) \right\|
               + 
               \sum_{k=1}^{n_2-1} \left\| \mathbf{H}_k^{\hat{f}}(t_j) \right\| \\
               & 
               +
               \left\| 
                   \sum_{k=1}^{n_1} \mathbf{H}^f_{k}(t_j)
                   +
                   \sum_{k=1}^{n_2-1} \mathbf{H}_k^{\hat{f}}(t_j)
               \right\|
             \bigg)
        \end{aligned}
         \\
        C = \left\| \mathbf{H}(t_j) \right\|
    \end{array}
\end{equation*}
\end{theorem}

Note that the above analysis reveals that we can choose any non-negative $\beta(t)$ as the control function. However, there is an implicit hazard that could potentially undermine the aforementioned analysis. For unconstrained combinatorial optimization, $|\psi(T) \rangle$ is invariably in the feasible domain regardless of the unitary operator employed. In contrast, for constrained combinatorial optimization, if we opt for any non-negative $\beta(t)$ and any $\mathbf{H}^{\hat{f}}_k$, $|\psi(T) \rangle$ may be not in feasible domain. Therefore, for constrained combinatorial optimization, an additional constraint comes here. The detailed descriptions can be found in Appendix \ref{appendix: discretization in lyapunov 1}.


\newtcolorbox{CircuitRow}[2][]{%
  enhanced,
  colback=white,
  colframe=#2!60!black,
  boxrule=0.5pt,
  arc=2mm,
  left=2mm,right=2mm,top=1.2mm,bottom=1.2mm,
  borderline west={2.2pt}{0pt}{#2!80!black},
  drop shadow={blur shadow, shadow xshift=1.2pt, shadow yshift=-1.2pt, shadow blur steps=8, opacity=0.25},
  before upper=\raggedright\setlength{\parindent}{0pt},
  #1
}

\newcommand{\CircuitScale}{0.90}
\newlength{\SubcircuitPadX}
\newlength{\SubcircuitPadY}
\pgfmathsetlength{\SubcircuitPadX}{6pt/\CircuitScale}
\pgfmathsetlength{\SubcircuitPadY}{6pt/\CircuitScale}

\tikzset{
  PBox/.style={
    dashed,
    rounded corners,
    line width=0.9pt,
    inner xsep=\SubcircuitPadX,
    inner ysep=\SubcircuitPadY,
  },
  PLabel/.style={font=\scriptsize, anchor=south},
  MeasLabel/.style={font=\scriptsize, anchor=south},
}

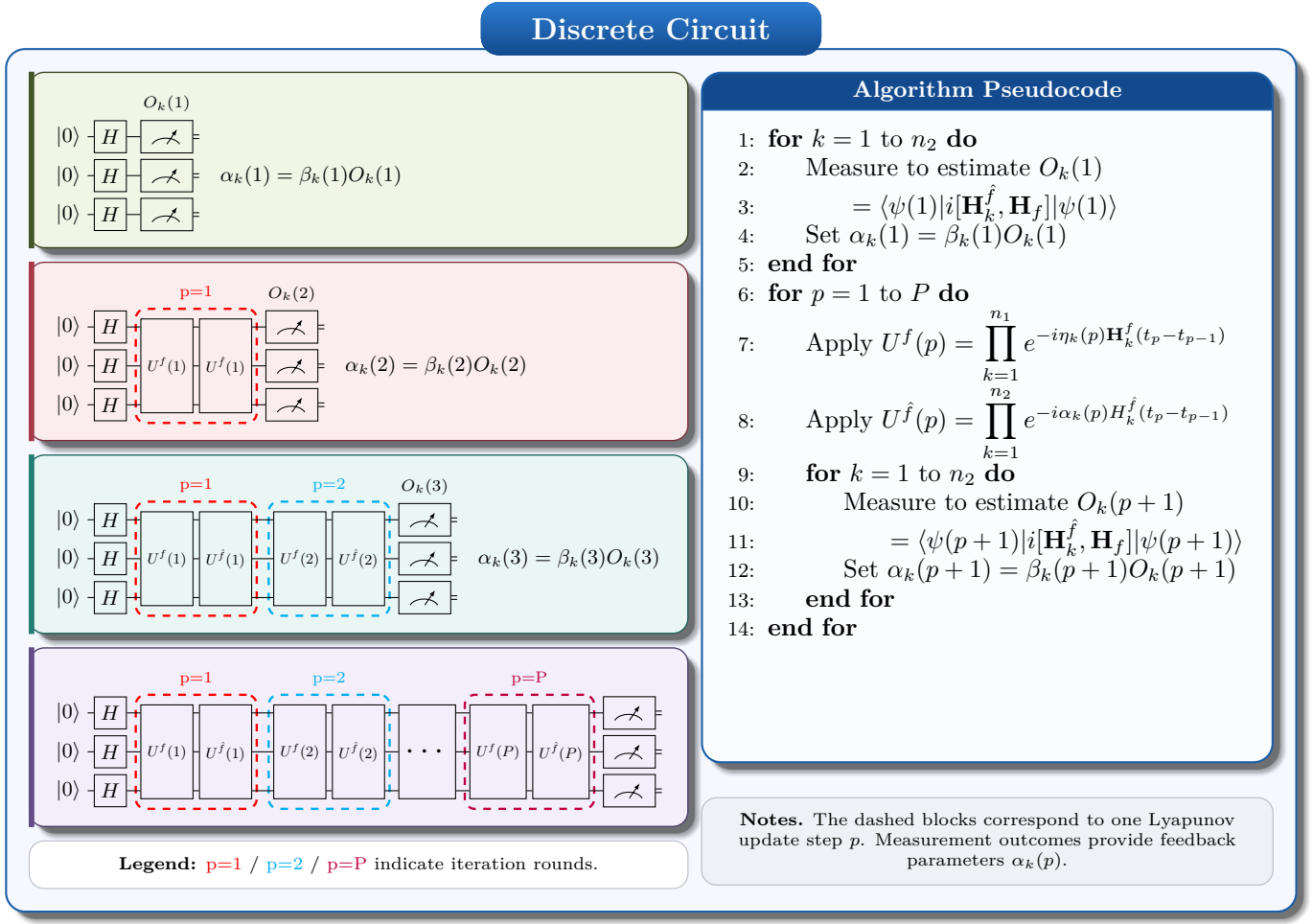
\begin{figure*}
\centering

\definecolor{MainFrame}{HTML}{1B5FA7}
\definecolor{MainBack}{HTML}{F5F8FF}
\definecolor{TitleTop}{HTML}{2D7DD2}
\definecolor{TitleBot}{HTML}{134A8E}

\definecolor{PaperWhite}{HTML}{FFFFFF}
\definecolor{SoftGray}{HTML}{EEF1F5}
\definecolor{CodeBack}{HTML}{FBFCFE}
\definecolor{LineGray}{HTML}{4B5563}

\definecolor{RowOlive}{HTML}{556B2F}
\definecolor{RowOliveTint}{HTML}{F1F5E6}
\definecolor{RowRed}{HTML}{D1495B}
\definecolor{RowRedTint}{HTML}{FBEDEE}
\definecolor{RowTeal}{HTML}{2A9D8F}
\definecolor{RowTealTint}{HTML}{EAF6F4}
\definecolor{RowPurple}{HTML}{7B5EA7}
\definecolor{RowPurpleTint}{HTML}{F3EEF9}

\begin{tcolorbox}[
    enhanced,
    width=\textwidth,
    colback=MainBack,
    colframe=MainFrame,
    boxrule=0.8pt,
    arc=2.5mm,
    left=2mm,right=2mm,top=2mm,bottom=2mm,
    title=\textbf{Discrete Circuit},
    fonttitle=\bfseries\large,
    colbacktitle=TitleBot,
    coltitle=white,
    boxed title style={
      boxrule=0pt,
      arc=2.5mm,
      left=6mm,right=6mm,top=1.2mm,bottom=1.2mm,
      interior style={top color=TitleTop, bottom color=TitleBot}
    },
    attach boxed title to top center={yshift=-1mm},
    drop shadow={blur shadow, shadow xshift=1.2pt, shadow yshift=-1.2pt, shadow blur steps=8, opacity=0.25}
]

\noindent
\begin{minipage}[t]{0.53\textwidth}
  \vspace{0pt}

  \begin{CircuitRow}[colback=RowOliveTint]{RowOlive}
  \noindent\begin{tikzpicture}[scale=\CircuitScale,transform shape]
  \begin{yquant}
      qubit {$\ket{0}$} q[3];
      h q;
      [value=1.2mm] hspace q;
      [name=MeasOne] measure q;
      output {~ $\alpha_k(1) = \beta_k(1) O_k(1)$} q[1];
  \end{yquant}

  \node[MeasLabel] at ([yshift=0.25mm]MeasOne-0.north) {$O_k(1)$};
  \end{tikzpicture}
  \end{CircuitRow}

  \begin{CircuitRow}[colback=RowRedTint]{RowRed}
  \noindent\begin{tikzpicture}[scale=\CircuitScale,transform shape]
  \begin{yquant}
      qubit {$\ket{0}$} q[3];
      h q;
      [value=1.2mm] hspace q;

      [name=UfOne]  box {\scalebox{0.72}{$U^f(1)$}} (q);
      [name=UhfOne] box {\scalebox{0.72}{$U^{\hat{f}}(1)$}} (q);

      [value=1.2mm] hspace q;
      [name=MeasTwo] measure q;
      output {~ $\alpha_k(2) = \beta_k(2) O_k(2)$} q[1];
  \end{yquant}

    \tikzset{
      PBox/.style={
        dashed, rounded corners,
        line width=0.9pt,
        inner xsep=2.0pt,
        inner ysep=4.0pt,
        outer sep=0pt
      }
    }

  \begin{scope}[on background layer]
    \node[draw=red, PBox, fit=(UfOne-0) (UhfOne-0)] (Pone) {};
  \end{scope}
  \node[red, PLabel] at ([yshift=0.35mm]Pone.north) {p=1};

  \node[MeasLabel] at ([yshift=0.25mm]MeasTwo-0.north) {$O_k(2)$};
  \end{tikzpicture}
  \end{CircuitRow}

  \begin{CircuitRow}[colback=RowTealTint]{RowTeal}
  \noindent\begin{tikzpicture}[scale=\CircuitScale,transform shape]
  \begin{yquant}
      qubit {$\ket{0}$} q[3];
      h q;

      [value=1.2mm] hspace q;

      [name=UfOne]  box {\scalebox{0.72}{$U^f(1)$}} (q);
      [name=UhfOne] box {\scalebox{0.72}{$U^{\hat{f}}(1)$}} (q);

      [value=2.4mm] hspace q;

      [name=UfTwo]  box {\scalebox{0.72}{$U^f(2)$}} (q);
      [name=UhfTwo] box {\scalebox{0.72}{$U^{\hat{f}}(2)$}} (q);

      [value=1.2mm] hspace q;
      
      [name=MeasThree] measure q;
      output {~ $\alpha_k(3) = \beta_k(3) O_k(3)$} q[1];
  \end{yquant}

    \tikzset{
      PBox/.style={
        dashed, rounded corners,
        line width=0.9pt,
        inner xsep=2.0pt,
        inner ysep=4.0pt,
        outer sep=0pt
      }
    }

  \begin{scope}[on background layer]
    \node[draw=red,  PBox, fit=(UfOne-0) (UhfOne-0)] (Pone) {};
    \node[draw=cyan, PBox, fit=(UfTwo-0) (UhfTwo-0)] (Ptwo) {};
  \end{scope}
  \node[red,  PLabel] at ([yshift=0.35mm]Pone.north) {p=1};
  \node[cyan, PLabel] at ([yshift=0.35mm]Ptwo.north) {p=2};

  \node[MeasLabel] at ([yshift=0.25mm]MeasThree-0.north) {$O_k(3)$};
  \end{tikzpicture}
  \end{CircuitRow}

  \begin{CircuitRow}[colback=RowPurpleTint]{RowPurple}
  \noindent\begin{tikzpicture}[scale=\CircuitScale,transform shape]
  \begin{yquant}
      qubit {$\ket{0}$} q[3];
      h q;

      [value=1.2mm] hspace q;

      [name=UfOne]  box {\scalebox{0.72}{$U^f(1)$}} (q);
      [name=UhfOne] box {\scalebox{0.72}{$U^{\hat{f}}(1)$}} (q);

      [value=2.4mm] hspace q;

      [name=UfTwo]  box {\scalebox{0.72}{$U^f(2)$}} (q);
      [name=UhfTwo] box {\scalebox{0.72}{$U^{\hat{f}}(2)$}} (q);

      [value=1.2mm] hspace q;

      box {\Large $\cdots$} (q);

      [value=1.2mm] hspace q;

      [name=UfP]   box {\scalebox{0.72}{$U^f(P)$}} (q);
      [name=UhfP]  box {\scalebox{0.72}{$U^{\hat{f}}(P)$}} (q);

      [value=1.2mm] hspace q;

      measure q;
      output {} q[1];
  \end{yquant}

    \tikzset{
          PBox/.style={
            dashed, rounded corners,
            line width=0.9pt,
            inner xsep=2.0pt,
            inner ysep=4.0pt,
            outer sep=0pt
          }
        }

  \begin{scope}[on background layer]
    \node[draw=red,    PBox, fit=(UfOne-0) (UhfOne-0)] (Pone) {};
    \node[draw=cyan,   PBox, fit=(UfTwo-0) (UhfTwo-0)] (Ptwo) {};
    \node[draw=purple, PBox, fit=(UfP-0)   (UhfP-0)]   (Pp)   {};
  \end{scope}
  \node[red,    PLabel] at ([yshift=0.35mm]Pone.north) {p=1};
  \node[cyan,   PLabel] at ([yshift=0.35mm]Ptwo.north) {p=2};
  \node[purple, PLabel] at ([yshift=0.35mm]Pp.north) {p=P};

  \end{tikzpicture}
  \end{CircuitRow}

  \begin{tcolorbox}[enhanced, colback=PaperWhite, colframe=LineGray!40, boxrule=0.4pt,
    arc=2mm, left=1.5mm,right=1.5mm,top=1mm,bottom=1mm]
    \scriptsize
    \textbf{Legend:} \textcolor{red}{p=1} / \textcolor{cyan}{p=2} / \textcolor{purple}{p=P} indicate iteration rounds.
  \end{tcolorbox}

\end{minipage}
\hfill
\begin{minipage}[t]{0.46\textwidth}
  \vspace{0pt}

  \begin{tcolorbox}[
        enhanced,
        title=\textbf{Algorithm Pseudocode},
        colback=CodeBack,
        colframe=MainFrame,
        colbacktitle=TitleBot,
        coltitle=white,
        fonttitle=\bfseries,
        boxrule=0.8pt,
        arc=2.5mm,
        left=2mm,right=2mm,top=1.5mm,bottom=1.5mm,
        drop shadow={blur shadow, shadow xshift=1.2pt, shadow yshift=-1.2pt, shadow blur steps=8, opacity=0.25},
        height=9.6cm,
        valign=top,
        boxed title style={
            colback=TitleBot,
            interior style={top color=TitleTop, bottom color=TitleBot}
        },
    ]
{\normalsize
\begin{algorithmic}[1]
    \For{$k = 1$ to $n_2$}
        \State Measure to estimate $O_k(1)$
        \State \hspace*{\algorithmicindent}
        $\displaystyle = \langle \psi(1)|
        i [\mathbf{H}^{\hat{f}}_k, \mathbf{H}_f]
        | \psi(1) \rangle$
        \State Set $\alpha_k(1) = \beta_k(1) O_k(1)$
    \EndFor
    \For{$p = 1$ to $P$}
    \State Apply $U^f(p) = $
      $\displaystyle \prod_{k=1}^{n_1} e^{-i\eta_k(p)\mathbf{H}^f_k (t_p-t_{p-1})}$
    \State Apply $U^{\hat{f}}(p) = $
      $\displaystyle \prod_{k=1}^{n_2} e^{-i\alpha_k(p) H_k^{\hat{f}} (t_p-t_{p-1})}$
    \For{$k = 1$ to $n_2$}
        \State Measure to estimate $O_k(p+1)$
        \State \hspace*{\algorithmicindent}
        $\displaystyle = \langle \psi(p+1)|
        i [\mathbf{H}^{\hat{f}}_k, \mathbf{H}_f]
        | \psi(p+1) \rangle$
        \State Set $\alpha_k(p+1) = \beta_k(p+1) O_k(p+1)$
    \EndFor
  \EndFor
\end{algorithmic}
}
  \end{tcolorbox}

  \vspace{1mm}
  \begin{tcolorbox}[enhanced, colback=SoftGray, colframe=LineGray!35,  boxrule=0.5pt,
    arc=2mm, left=2mm,right=2mm,top=1mm,bottom=1mm]
    \scriptsize
    \textbf{Notes.} The dashed blocks correspond to one Lyapunov update step $p$.
    Measurement outcomes provide feedback parameters $\alpha_k(p)$.
  \end{tcolorbox}

\end{minipage}

\end{tcolorbox}

\caption{Discrete Implementation of the Algorithm Derived from the Lyapunov Framework}
\label{fig:algorithm}
\end{figure*}

\subsection{Lyapunov Function with Two Parameters}
The above analysis represents the approximation ratio using a parameter $\lambda (t) \ge 0$, which can also be expressed as a ratio of two positive parameters, i.e., $\lambda(t) = \frac{y(t)}{x(t)}$ for $x(t),y(t) \in \mathbb{R}_{+}$. Accordingly, we introduce a Lyapunov function with two parameters, which takes the following form:
\begin{equation}
    E(t) = x(t) \langle \psi(t)| \mathbf{H}_f | \psi(t) \rangle- y(t) \langle \psi^*| \mathbf{H}_f | \psi^* \rangle
\end{equation}    
where $x(t),y(t)$ are non-decreasing. The additional parameter potentially enables a more precise estimation of the approximation ratio during the execution of the algorithm. Similarly, if it can be ensured that $E(t) \geq E(0)$ holds for any $t$, the inequality can be reformulated to derive an expression for the approximation ratio:
\begin{equation}
    \begin{aligned}
        \langle \psi(t)| \mathbf{H}_f | \psi(t) \rangle \ge & \cfrac{y(t)-y(0)}{x(t)} \langle \psi^*| \mathbf{H}_f | \psi^* \rangle \\
        & + \cfrac{x(0)}{x(t)}\langle \psi(0)| \mathbf{H}_f | \psi(0) \rangle
    \end{aligned}
    \label{eq::10}
\end{equation}
The condition $\frac{dE(t)}{dt} \ge 0$ remains a sufficient criterion to ensure $E(t) \ge E(0)$. The derivative of the two-parameter Lyapunov function is given by:
\begin{equation}
    \begin{aligned}
        \cfrac{d E(t)}{d t} = & x(t) \langle \psi(t)| i [\mathbf{H}^{\hat{f}}(t) , \mathbf{H}_f] | \psi(t) \rangle \\
         + &\cfrac{d x(t)}{d t} \langle \psi(t) |  \mathbf{H}_f | \psi(t) \rangle - \cfrac{d y(t)}{dt} \langle \psi^*| \mathbf{H}_f | \psi^* \rangle
    \end{aligned}
\end{equation}

In addition, on some combinatorial optimization problems, the upper bound of the optimal solution can also be characterized as the summation of two parts: the function value of the current solution and an adjustment term. That implies that, within the two-parameter framework, the upper bound of the optimal solution $\langle \psi^* | \mathbf{H}_f | \psi^* \rangle$ takes the following form:
\begin{equation}\label{eq::two_para_optimal_solution}
        \begin{aligned}
            0 &\le \langle \psi^*| \mathbf{H}_f | \psi^* \rangle \\
            & \le a(t) \langle \psi(t)| \mathbf{H}_f | \psi(t) \rangle + b(t) \langle \psi(t)| \mathbf{Q}(t) | \psi(t) \rangle
        \end{aligned}
\end{equation}
where $a(t),b(t)$ are non-negative. Therefore, we obtain a sufficient condition such that $ \frac{dE(t)}{dt} $ is non-negative, from which a set of differential equations governing the parameters $x(t)$ and $y(t)$ can be derived:
\begin{widetext}
\begin{equation*}
    \begin{array}{c}
            \cfrac{d E(t)}{d t}\ge \left(x(t) \langle \psi(t)| i [\mathbf{H}^{\hat{f}}(t) , \mathbf{H}_f] | \psi(t) \rangle - b(t) \cfrac{d y(t)}{d t} \langle \psi(t)| \mathbf{Q}(t) | \psi(t) \rangle \right)+ \left( \cfrac{d x(t)}{d t} - a(t) \cfrac{d y(t)}{d t} \right) \langle \psi(t) |  \mathbf{H}_f | \psi(t) \rangle = 0\\ 
        \Downarrow\\
        \begin{cases}
            x(t) \langle \psi(t)| i [\mathbf{H}^{\hat{f}}(t) , \mathbf{H}_f] | \psi(t) \rangle = b(t) \cfrac{d y(t)}{d t} \langle \psi(t)| \mathbf{Q}(t) | \psi(t) \rangle \\
            \cfrac{d x(t)}{d t} - a(t) \cfrac{d y(t)}{d t} = 0
        \end{cases} 
    \end{array}
\end{equation*}
\end{widetext}
Similar with the one parameter Lyapunov framework, the Hamiltonian $ \mathbf{H}(t) $ can be decomposed into two components: one commutes with the target Hamiltonian and the other one does not. Here, we continue to adopt parameters defined in the last subsection. From the above equations, the explicit forms of parameters $x(t)$ and $y(t)$ can be obtained as follows:
\begin{equation}\label{eq::two_para_estimate}
    \begin{aligned}
        y(t) & = y(0) + \int_0^t \cfrac{x(s)}{b(s)} \cdot \cfrac{\sum_{k=1}^{n_2}\alpha_k(s)O_k(s)} { \langle \psi(s)| \mathbf{Q}(s) | \psi(s) \rangle} ds \\
        x(t) & = x(0)\exp{\left[\int_0^t \cfrac{a(s)}{b(s)} \cdot \cfrac{\sum_{k=1}^{n_2}{\alpha_k(s) O_k(s)}}{\langle \psi(s)| \mathbf{Q}(s) | \psi(s) \rangle } ds \right]}
    \end{aligned}
\end{equation}
Substituting this into Eq. (\ref{eq::10}), we obtain the approximation ratio under the two parameter Lyapunov framework:
\begin{equation}\label{eq::two_para_ratio}
    \cfrac{y(T)-y(0)}{x(T)} = \cfrac{\int_0^T \frac{x(t)}{b(t)} \cdot \frac{\sum_{k=1}^{n_2}\beta_k(t)O^2_k(t)} { \langle \psi(t)| \mathbf{Q}(t) | \psi(t) \rangle} dt}{x(0)\exp{\left[\int_0^T \frac{a(t)}{b(t)} \cdot \frac{\sum_{k=1}^{n_2}\beta_k(t)O^2_k(t)} { \langle \psi(t)| \mathbf{Q}(t) | \psi(t) \rangle} dt \right]}}
\end{equation}
By appropriately adjusting the tunable parameter $\beta(t)$, the above approximation ratio could be large. Within the two parameter Lyapunov framework, we also analyze the increment in the approximation ratio achieved during each iteration of the discretized algorithm, which is formally as follows. The detailed analysis is presented in Appendix \ref{Appendix::Two-Parameter-Discretion}.


\begin{theorem}
If the evolution time $t_{j+1}-t_j \le \frac{\epsilon}{Bx(t_{j+1})+C\epsilon+\sqrt{A x(t_{j+1})\epsilon}} $, the parameter increment in the $j+1$-th iteration can be expressed as below with the bounded discretization error $\epsilon$
    \begin{equation}
        \begin{array}{c}
            x(t_{j+1}) = x(t_j) \cdot \cfrac{b(t_j) \langle \psi(t_j)| \mathbf{Q}(t_j) | \psi(t_j)}{c(t_j)}\\
            y(t_{j+1}) = y(t_j) + x(t_j) \cdot \cfrac{\sum_{k=1}^{n_2} \beta_k(t_j)O^2_k(t_j) (t_{j+1}-t_j)}{c(t_j)}
        \end{array}
    \end{equation}
    where 
    \begin{equation*}
        \begin{aligned}
           c(t_j) = & b(t_j) \langle \psi(t_j)| \mathbf{Q}(t_j) | \psi(t_j) \\
        & -a(t_j) \sum_{k=1}^{n_2} \beta_k(t_j)O^2_k(t_j) (t_{j+1}-t_j)
        \end{aligned}
    \end{equation*}
\end{theorem}

\subsection{Application for Max-Cut Problem}
The Max-Cut problem is among the most common benchmark for combinatorial optimization problem solving by quantum algorithms. We solve the Max-Cut problem on unweighted, undirected graphs and demonstrate the performance of the approximation ratio on our framework in comparison with the actual approximation ratio.

Recall the definition of the Max-Cut problem. For a graph $G = (V, E)$, where $V$ is the set of vertices with $|V| = n$ and $E$ is the set of edges with $|E| = m$, finding a partition $(V_1,V_2)$ of the vertex set $V$ such that edges between $V_1$ and $V_2$ are maximized. The Hamiltonian of the Max-Cut problem is defined on $n$ qubits as:
\begin{equation*}
    \mathbf{H}_f = \sum_{<j,k>\in E} \cfrac{1}{2}(\mathbf{I}^{\otimes n} - \mathbf{Z}_j\mathbf{Z}_k)
\end{equation*}

According to Eq.(\ref{eq::Algorithm_Hamiltonian_Choice}), we choose $\mathbf{H}^f(t)$ and $\mathbf{H}^{\hat{f}}(t)$ as below:
\begin{equation}\label{eq::using_FALQON_Hamilton}
    \mathbf{H}^f(t) = \cfrac{\mathbf{H}_f}{m}, \quad \mathbf{H}^{\hat{f}}(t) = \beta(t)O(t)\sum_{j=1}^n\mathbf{X}_j
\end{equation}
where $O(t) = \bra{\psi(t)} i [\sum_{j=1}^n\mathbf{X}_j, \mathbf{H}_f] \ket{\psi(t)}$ denotes the observable used in the measurement-feedback control, and $\beta(t)$ is a tunable parameter. Therefore, the Hamiltonian of the algorithm can be defined as:
\begin{equation*}
    \mathbf{H}(t) = \cfrac{\mathbf{H}_f}{m} + \beta(t)\bra{\psi(t)} i [\sum_{j=1}^n\mathbf{X}_j, \mathbf{H}_f] \ket{\psi(t)}\sum_{j=1}^n\mathbf{X}_j
\end{equation*}
Note that the above choice for $\mathbf{H}^{f}(t)$ and $\mathbf{H}^{\hat{f}}(t)$ is a concrete instance of QAOA algorithms with the parameter $\beta(t)\bra{\psi(t)} i [\sum_{j=1}^n\mathbf{X}_j, \mathbf{H}_f] \ket{\psi(t)}$ for mixer operators and $1/m$ for phase operators.

To obtain the lower bound of the approximation ratio, we first construct an upper bound on the optimal solution and a tighter upper bound will result in a better lower bound of the approximation ratio. For the one parameter Lyapunov function, as a demonstration of deriving a lower bound on the approximation ratio of the algorithm, we note that the Max-Cut problem admits a simple upper bound $m$:
\begin{equation*}
    \bra{\psi^*} \mathbf{H}_f \ket{\psi^*} \le \bra{\psi(t)} m\mathbf{I} \ket{\psi(t)}
\end{equation*}
Leveraging Eq.(\ref{eq::single_para_estimate}), we obtain the explicit forms of the lower bound for the approximation ratio derived from the one parameter Lyapunov function: 
\begin{equation}
    \begin{array}{cc}
         \lambda(T)-\lambda(0) = \displaystyle\int_{0}^{T} \cfrac{\beta(s) \left[\langle \psi(s)| i [\sum_{k=1}^n\mathbf{X}_k, \mathbf{H}_f] | \psi(s) \rangle \right]^2} {\bra{\psi(s)} m\mathbf{I} \ket{\psi(s)}}ds
    \end{array}
\end{equation}

For two parameter Lyapunov function, we need to construct another upper bound on the optimal solution. To obtain an upper bound that is not much looser than the one parameter case, we impose that the right-hand side of Eq. (\ref{eq::two_para_optimal_solution}) be less than or equal to $\bra{\psi(t)} m\mathbf{I} \ket{\psi(t)}$, from which the upper bound of $\langle \psi(t)| \mathbf{Q}(t) | \psi(t) \rangle$ can be derived as:
\begin{equation*}
    \begin{array}{cc}
        \langle \psi(t)| \mathbf{Q}(t) | \psi(t) \rangle \leq \cfrac{\bra{\psi(t)} m\mathbf{I} \ket{\psi(t)} - a(t)\langle \psi(t)| \mathbf{H}_f | \psi(t) \rangle}{b(t)}
    \end{array}
\end{equation*}

To ensure the validity of Eq. (\ref{eq::two_para_optimal_solution}), one approach is to set $\langle \psi(t)| \mathbf{Q}(t) | \psi(t) \rangle$ at its upper bound and choose $a(t) = b(t) = 1$. Consequently, Eq. (\ref{eq::two_para_estimate}) can be reformulated as:

\begin{equation}
    \begin{array}{cc}
    \displaystyle
         y(T) = y(0) + \int_{0}^{T}  \cfrac{x(s)\beta(s) \left[\langle \psi(s)| i [\sum_{k}\mathbf{X}_k, \mathbf{H}_f] | \psi(s) \rangle \right]^2} {  \langle \psi(s)| m\mathbf{I}-\mathbf{H}_f | \psi(s) \rangle} ds \\
    \displaystyle
        x(T) = x(0)\exp{\left[\int_{0}^{T} \cfrac{\beta(s) \left[\langle \psi(s)| i [\sum_{k}\mathbf{X}_k, \mathbf{H}_f] | \psi(s) \rangle \right]^2}{ \langle \psi(s)| m\mathbf{I}-\mathbf{H}_f | \psi(s) \rangle} ds \right]}
    \end{array}
\end{equation}


During the execution of the algorithm at time $t$, the parameters $\lambda(t)$, $x(t)$, and $y(t)$ associated with the approximation ratio can be calculated using numerical integration methods such as the trapezoidal rule, which produces a lower bound of the true approximation ratio at any time $t$ via $\lambda(t) - \lambda(0)$ and $\frac{y(t) - y(0)}{x(t)}$.



\begin{figure*}
    \centering
    \begin{minipage}{0.95\linewidth}
        \includegraphics[width=1\linewidth]{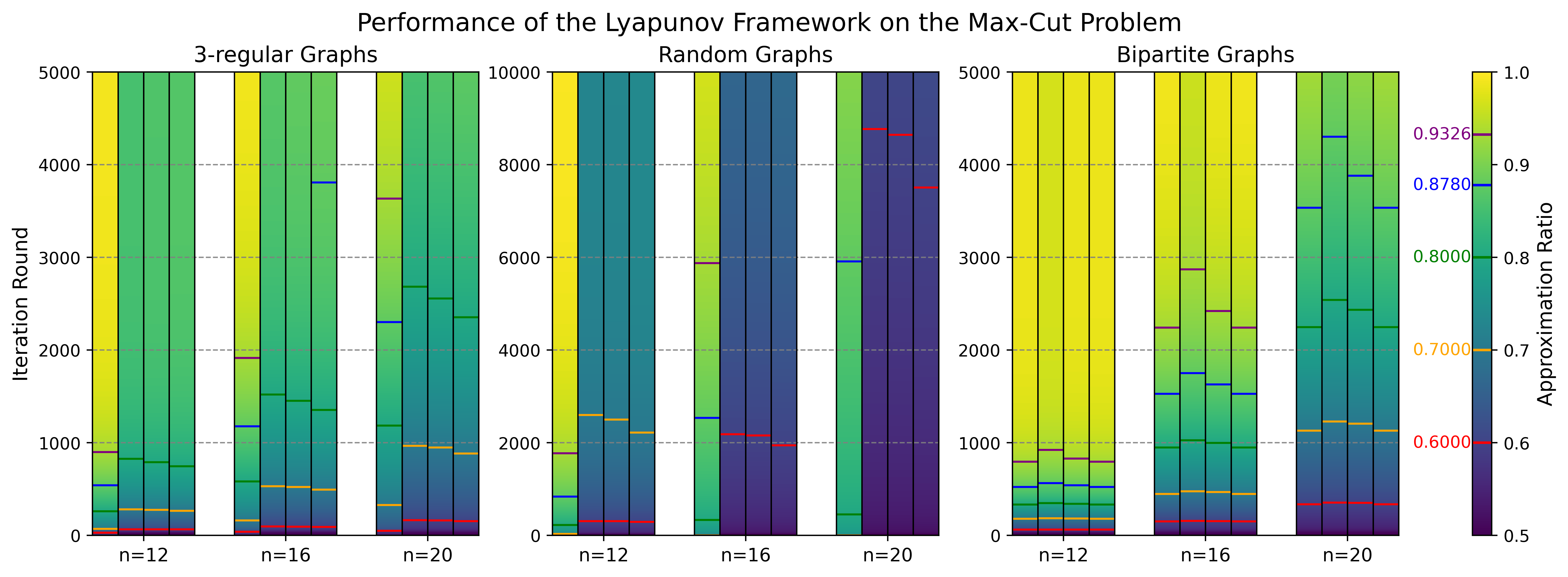}
    \end{minipage}
    
    \caption{The algorithm defined in Eq. (\ref{eq::using_FALQON_Hamilton}) was evaluated on nine experimental settings, corresponding to 3-regular graphs, Erdős-Rényi random graphs, and bipartite graphs with node sizes $n={12,16,20}$ (from left to right). Each experiment contains four bar plots, where the color of each bar encodes the value of the corresponding parameter across different iteration rounds. Within each experiment, the four bars from left to right represent, respectively,
    $\langle \psi(t_j)| \mathbf{H}_f | \psi(t_j) \rangle / \langle \psi^*| \mathbf{H}_f | \psi^* \rangle,\quad
    \lambda(t_j) - \lambda(0),\quad
    (y(t_j) - y(0)) / x(t_j),\quad
    \langle \psi(t_j)| \mathbf{H}_f | \psi(t_j) \rangle / m$.}
    \label{approximation-ratio}
\end{figure*}

\section{Numerical Experiments}
In our numerical experiments, we simulate our framework on three classes of graphs: 3-regular graphs, Erdős-Rényi random undirected graphs with $p = 0.5$, and binomial random bipartite graphs with $p = 0.5$. The test set includes all connected 3-regular graphs with $n \in \{4,6,8,...,20\}$, totaling 501 instances. For Erdős-Rényi random graphs and random bipartite graphs, we generate 50 connected graphs for each $n \in \{10,11,...,,20\}$. The algorithm is initialized in the state $|+\rangle^{\otimes n} = \left(\tfrac{1}{\sqrt{2}}(|0\rangle + |1\rangle)\right)^{\otimes n}$. The single-step size is set to $\Delta t = 0.08$, and the total number of iterations is fixed at $R = 10{,}000$. The parameter $\beta$ is defined as a gradually decelerating decreasing function,
\begin{equation}
    \beta(t) = 0.04 \left(0.5 \big(1 - e^{-\tfrac{2}{R}(\Delta t R - t)}\big) + 0.5\right)
    \label{beta_value}
\end{equation}

Experiments were performed on a workstation running Ubuntu 24.04.1 LTS, equipped with dual-socket Intel Xeon Gold 5222 CPUs @ 3.80 GHz (8 physical cores / 16 threads), 125 GiB RAM, and an NVIDIA GeForce RTX 3090 GPU (24 GB; driver 550.144.03). The code implementing our framework has been made publicly available at \url{https://gitee.com/lzier/lyapunov_qaoa_open}.

\subsection{Approximation Ratio} \label{subsection::5.1}
We implement the algorithm as defined in Eq.(\ref{eq::using_FALQON_Hamilton}) and evaluate both the true approximation ratio, i.e., $\bra{\psi(t_j)} \mathbf{H}_f \ket{\psi(t_j)} / \bra{\psi^*} \mathbf{H}_f \ket{\psi^*}$, during its execution and the lower bounds on the approximation ratio derived from the one parameter and two parameter Lyapunov functions, respectively. Here, $\bra{\psi^*}\mathbf{H}_f \ket{\psi^*}$ is the optimal value of the Max-Cut instance, obtained by brute-force enumeration over all $2^n$ bit-string configurations. To evaluate the effectiveness of the approximate ratio from continuous procedures to discrete procedures, we calculate the averages of $\lambda(t_j) - \lambda(0)$ and $\left( y(t_j) - y(0)\right)/x(t_j)$ on each test instance. In addition, we calculate the $\bra{\psi(t)} \textbf{H}_f \ket{\psi(t)}/m$ as references to illustrate errors from discrete procedure.

Fig. \ref{approximation-ratio} illustrates the performance of the algorithm. From the figure, we observe that as the number of iterations increases, the true approximation ratio gradually converges to a value close to 1, indicating that the algorithm is capable of producing near-optimal solutions when the number of rounds tends to infinity. Moreover, for any finite number of iterations, the approximation ratio derived from both the one parameter and two parameter Lyapunov functions consistently lie below the true ratio, thereby serving as guaranteed lower bounds. Because a relatively loose upper bound on the optimal solution was used, the Lyapunov-based approximation ratio deviate from the true approximation ratio. Nevertheless, the numerical results demonstrate that they still provide accurate approximations to $\bra{\psi(t_j)}\mathbf{H}_f\ket{\psi(t_j)}/m$. In practice, if a tighter upper bound for the optimal solution is available, we expect the approximation ratio derived from the Lyapunov framework to achieve even stronger performance.


Erdős-Rényi random graphs tend to exhibit higher connectivity, and their degree distribution is less uniform than that of 3-regular graphs. As a consequence, the local structural variations across different regions of the graph are substantial, which makes it difficult to characterize the full evolution using a single set of global parameters. This leads to a slower convergence of the true approximation ratio compared with 3-regular graphs, as reflected by the fact that Erdős-Rényi graphs attain a lower approximation ratio under the same number of iterations, particularly for larger values of $n$. These structural differences also influence the lower bound for the approximation ratio derived from the one parameter and two parameter Lyapunov formulations, resulting in a larger gap for the true approximation ratio. Nevertheless, both discrete procedures serve as good approximations to $\bra{\psi(t_j)} \mathbf{H}_f \ket{\psi(t_j)} / m$.


Since the optimal solution for bipartite graphs is exactly equal to $ m $, the true approximation ratio matches $\langle \psi(t_j) | \mathbf{H}_f | \psi(t_j) \rangle / m$ at every iteration step. Because the lower bounds obtained from both Lyapunov formulations provide accurate estimates of $\langle \psi(t_j) | \mathbf{H}_f | \psi(t_j) \rangle / m$, the resulting approximation-ratio estimates are more precise for bipartite graphs than for 3-regular or Erdős-Rényi random graphs. Moreover, across all tests, the two parameter formulation consistently outperforms the one parameter.



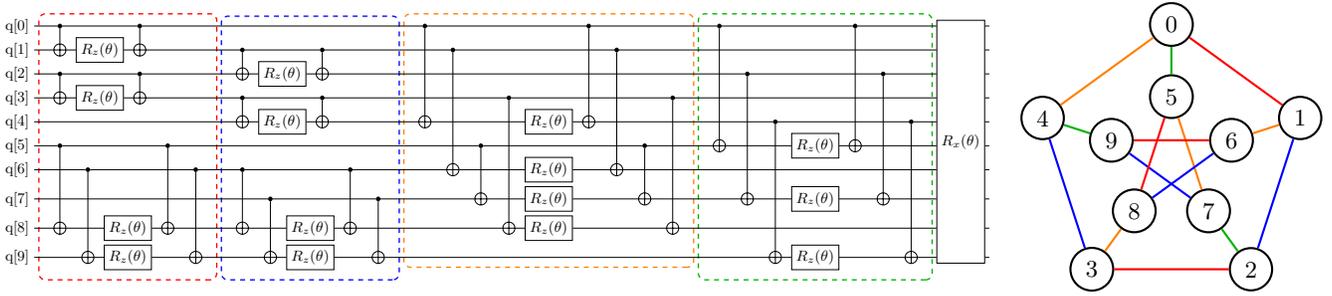
\begin{figure*}
    \begin{minipage}{0.74\linewidth}
    \resizebox{1\linewidth}{!}{
        \begin{tikzpicture}
            \begin{yquant}
            
                qubit q[10];
            
                [this subcircuit box style={
                draw=red,        
                dashed,          
                rounded corners, 
                line width=0.8pt,
                inner ysep=6pt,  
            }]
            subcircuit {
                qubit {} q[10];
        
                cnot q[1] | q[0];
                cnot q[3] | q[2];
                cnot q[8] | q[5];
                cnot q[9] | q[6];
                
                box {$R_z{(\theta)}$} (q[1]);
                box {$R_z{(\theta)}$} (q[3]);
                box {$R_z{(\theta)}$} (q[8]);
                box {$R_z{(\theta)}$} (q[9]);
                
                cnot q[1] | q[0];
                cnot q[3] | q[2];
                cnot q[8] | q[5];
                cnot q[9] | q[6];
            } (q);
            [this subcircuit box style={
                draw=blue,        
                dashed,          
                rounded corners, 
                line width=0.8pt,
                inner ysep=6pt,  
            }]
            subcircuit {
                qubit {} q[10];
        
                cnot q[2] | q[1];
                cnot q[4] | q[3];
                cnot q[8] | q[6];
                cnot q[9] | q[7];
                
                box {$R_z{(\theta)}$} (q[2]);
                box {$R_z{(\theta)}$} (q[4]);
                box {$R_z{(\theta)}$} (q[8]);
                box {$R_z{(\theta)}$} (q[9]);
                
                cnot q[2] | q[1];
                cnot q[4] | q[3];
                cnot q[8] | q[6];
                cnot q[9] | q[7];
            } (q);
            [this subcircuit box style={
                draw=orange,        
                dashed,          
                rounded corners, 
                line width=0.8pt,
                inner ysep=6pt,  
            }]
            subcircuit {
                qubit {} q[10];
        
                cnot q[4] | q[0];
                cnot q[6] | q[1];
                cnot q[7] | q[5];
                cnot q[8] | q[3];
                
                box {$R_z{(\theta)}$} (q[4]);
                box {$R_z{(\theta)}$} (q[6]);
                box {$R_z{(\theta)}$} (q[7]);
                box {$R_z{(\theta)}$} (q[8]);
                
                cnot q[4] | q[0];
                cnot q[6] | q[1];
                cnot q[7] | q[5];
                cnot q[8] | q[3];
            } (q);
            [this subcircuit box style={
                draw=green!70!black,        
                dashed,          
                rounded corners, 
                line width=0.8pt,
                inner ysep=6pt,  
            }]
            subcircuit {
                qubit {} q[10];
        
                cnot q[5] | q[0];
                cnot q[7] | q[2];
                cnot q[9] | q[4];
                
                box {$R_z{(\theta)}$} (q[5]);
                box {$R_z{(\theta)}$} (q[7]);
                box {$R_z{(\theta)}$} (q[9]);
                
                cnot q[5] | q[0];
                cnot q[7] | q[2];
                cnot q[9] | q[4];
            } (q);
            
            box {$R_x{(\theta)}$} (q);
             \end{yquant}
        \end{tikzpicture}
        }
    \end{minipage}
    \begin{minipage}{0.25\linewidth}
        \begin{tikzpicture}[every node/.style={circle, draw, fill=white, minimum size=5mm},
                    scale=1.2, thick]

            \node (A1) at (90:1.5) {0};
            \node (A2) at (18:1.5) {1};
            \node (A3) at (-54:1.5) {2};
            \node (A4) at (-126:1.5) {3};
            \node (A5) at (162:1.5) {4};
            
            \node (B1) at (90:0.7) {5};
            \node (B2) at (18:0.7) {6};
            \node (B3) at (-54:0.7) {7};
            \node (B4) at (-126:0.7) {8};
            \node (B5) at (162:0.7) {9};
            
            \draw[red] (A1)--(A2);
            \draw[blue] (A2)--(A3);
            \draw[red] (A3)--(A4);
            \draw[blue] (A4)--(A5);
            \draw[orange] (A5)--(A1);
            
            \draw[orange] (B1)--(B3);
            \draw[blue] (B3)--(B5);
            \draw[red] (B5)--(B2);
            \draw[blue] (B2)--(B4);
            \draw[red] (B4)--(B1);
            
            \draw[green!70!black] (A1)--(B1);
            \draw[orange] (A2)--(B2);
            \draw[green!70!black] (A3)--(B3);
            \draw[orange] (A4)--(B4);
            \draw[green!70!black] (A5)--(B5);
        
        \end{tikzpicture}
    \end{minipage}
    \caption{The right panel illustrates one possible edge-coloring scheme of the Petersen graph based on Vizing's theorem, while the left panel depicts a single-layer iterative circuit of the algorithm defined in Eq. (\ref{eq::using_FALQON_Hamilton}). Edges sharing the same color can be processed in parallel, resulting in a circuit of constant depth.}
    \label{Fig::Peterson_example}
\end{figure*}

\begin{figure*}
    \centering
    \begin{minipage}{0.95\linewidth}
        \includegraphics[width=1\linewidth]{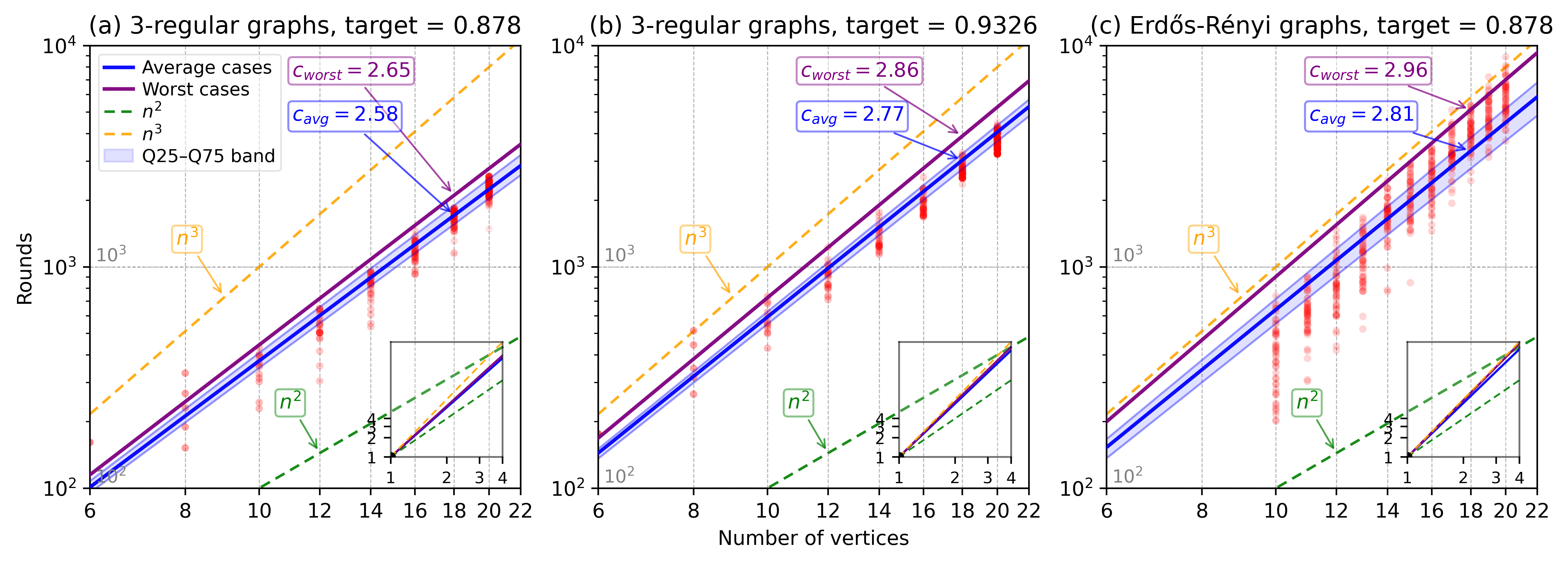}
    \end{minipage}
    \caption{The scatter points represent the number of iterations required for the algorithm to achieve an approximation ratio on instances. The blue line denotes the linear fit of these data points. The purple line corresponds to the linear fit of the largest iteration counts observed for each graph size $n$. The light-blue shaded region shows the interquartile spread of the iteration counts (25th–75th percentiles) for each graph size $n$. For comparison, the yellow and green lines represent the functions $R(n) = n^3$ and $R(n) = n^2$, respectively.
}
    \label{convergence rate}
\end{figure*}

\subsection{Convergence Rate}
We now examine the convergence rate of the algorithm defined in Eq. (\ref{eq::using_FALQON_Hamilton}). The convergence rate is characterized by the function $R(n)$, which denotes the number of iterations required to achieve a target approximation ratio for an undirected graph of size $n$. 
According to Vizing's theorem, any 3-regular graph can be edge-colored using at most four colors. Edges assigned the same color can be processed simultaneously, as they correspond to the commuting terms in Eq. (\ref{eq::using_FALQON_Hamilton}). Consequently, the circuit implementing a single iteration of the algorithm defined by Eq. (\ref{eq::using_FALQON_Hamilton}) has constant depth when applied to 3-regular graphs and the convergence rate $R(n)$ is expected to be a polynomial function of $n$, that is, there exists a constant $c$ such that $R(n) = O(n^c)$. An illustrative example is shown in Fig. \ref{Fig::Peterson_example}.


The experimental data set includes all 501 3-regular graphs with fewer than 20 vertices, together with 50 random undirected graphs for each size $n \in \{10,11,\ldots,20\}$. The best-known classical algorithm for solving the Max-Cut problem on general undirected graphs without special structural constraints achieves an approximation ratio of 0.878 \cite{GoemansWilliamson1995}, while for 3-regular graphs, the best classical approximation ratio is 0.9326 \cite{HalperinLivnatZwick2002}. Therefore, to estimate the exponent $c$, we record the number of iterations required for each instance to reach approximation ratios of 0.878 and 0.9326, respectively. These iteration counts are plotted against $n$ on a log-log scale of base 10, with the horizontal axis representing the size of the graph $n$ and the vertical axis representing the number of iterations. Linear regression is then performed on all data points (blue line) and on the maximal iteration counts for each value of $n$ (purple line). For reference, we also plot the curves $R(n)=n^2$ (green) and $R(n)=n^3$ (orange). The light-blue shaded band indicates the interquartile range (Q25–Q75) of the iteration counts for each graph size $n$. For each $n$, we compute the 25th and 75th percentiles of the required rounds across instances, and the two boundary curves are obtained by applying a linear fit to these percentile values. The results are shown in Fig. \ref{convergence rate}.


The left one in Fig. \ref{convergence rate} presents the iteration complexity of the algorithm on 3-regular graphs with a target approximation ratio of 0.878. The fitted scaling exponents are $c = 2.58$ on average and $c = 2.65$ in the worst case. In the log-log plot, most of the data points cluster closely around the regression line (blue), indicating that the algorithm exhibits stable performance in graphs of similar structure, without rare outliers with significantly worse behavior. 

The middle one in Fig. \ref{convergence rate} shows the number of iterations required to achieve an approximation ratio of 0.9326 on 3-regular graphs. The fitted scaling exponents are $c = 2.77$ on average and $c = 2.86$ in the worst case. Compared with the left one in Fig. \ref{convergence rate}, achieving a higher approximation ratio requires more iterations, leading to larger values of $c$. For reference, the best known classical approximation algorithm that attains a ratio of 0.9326 run in time $O(n^{3.5}/\log n)$\cite{HalperinLivnatZwick2002}. In contrast, the observed scaling suggests that even in the worst case, the circuit depth required by the algorithm defined in Eq. (\ref{eq::using_FALQON_Hamilton}) grows no faster than $O(n^3)$, highlighting the possible advantage of the quantum approximation algorithm over classical methods.

The right one in Fig. \ref{convergence rate} reports the iteration counts for achieving an approximation ratio of 0.878 on random undirected graphs. The fitted scaling exponents are $c = 2.81$ on average and $c = 2.96$ in the worst case. The larger exponents compared with the 3-regular case indicate that random graphs require deeper circuits to achieve comparable performance, underscoring the role of graph structure in algorithmic behavior.

\subsection{Comparison of Another Ansatz}
Next, we discuss the performance of the Lyapunov framework on another ansatz, in order to further verify that our framework can provide a reliable lower bound of approximation ratios for different variational quantum algorithms. The light-cone variational quantum algorithm (light-cone VQA) \cite{Wang2025} is an efficient VQA which, unlike the traditional QAOA, does not employ mixing and separation Hamiltonians. Instead, it orients the undirected edges of the original graph sequentially using a breadth-first search, and then constructs a circuit as a sequence of $Z\!Y$-Pauli gates arranged according to the resulting topological order.

In our experiments, we combine the light-cone VQA circuit with the idea of measurement feedback. Specifically, for the commuting terms we set $\mathbf{H}^f = \mathbf{0}$, and for the non-commuting terms, if the position of the $j$-th node in the BFS order is denoted by $seq_j$, then we set
\begin{equation}
    \mathbf{H}^{\hat{f}} = \beta(t)\sum_{<j,k>\in E \wedge seq_j < seq_k} Y_jZ_k
    \label{light-cone_Hamilton}
\end{equation}
and $\beta(t)$ is defined identically to Eq. (\ref{beta_value}). As before, the initial state is set to $\ket{+}$, the step size is chosen as $\Delta t = 0.08$, and the total number of iterations is $R = 30$. Throughout the execution, we calculate $\lambda(t_j)-\lambda(0)$, $\left( y(t_j) - y(0)\right)/x(t_j)$, and $\langle \psi(t)|\mathbf{H}_f|\psi(t)\rangle / m$, which serve to evaluate the performance of our framework.
\begin{figure}
    \centering
    \includegraphics[width=1\linewidth]{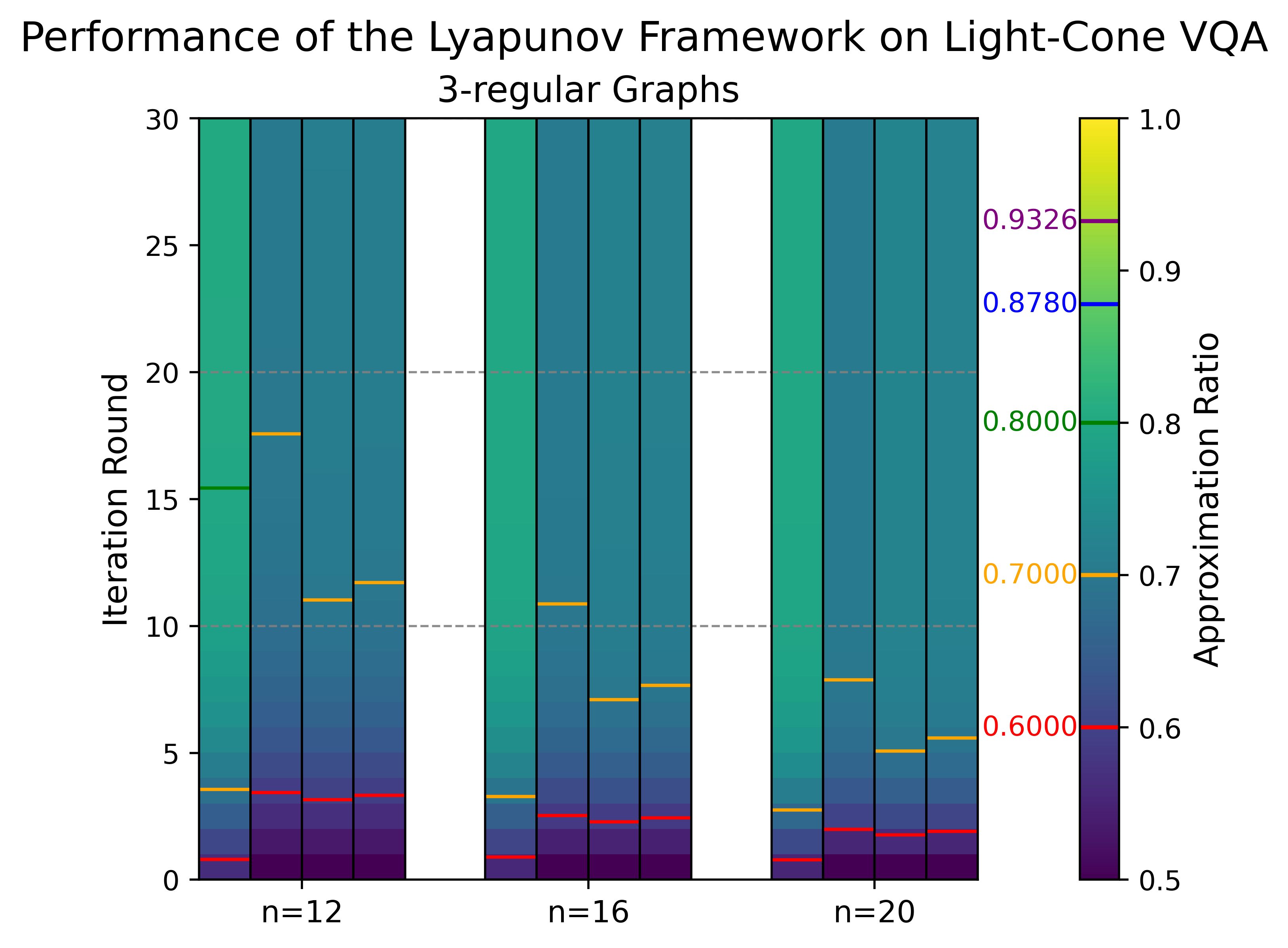}
    \caption{The results of running the algorithm defined in Eq. (\ref{light-cone_Hamilton}) on 3-regular graphs with $n = \{12, 16, 20\}$ are presented. For different $n$, the four bars from left to right represent, respectively,
    $\langle \psi(t_j)| \mathbf{H}_f | \psi(t_j) \rangle / \langle \psi^*| \mathbf{H}_f | \psi^* \rangle,\quad
    \lambda(t_j) - \lambda(0),\quad
    (y(t_j) - y(0)) / x(t_j),\quad
    \langle \psi(t_j)| \mathbf{H}_f | \psi(t_j) \rangle / m$.}
    \label{light-cone_3-regular}
\end{figure}
The experimental results are shown in Fig.\ref{light-cone_3-regular}, with legends consistent with those in Fig.\ref{approximation-ratio}. As observed, compared to the traditional QAOA, the Hamiltonian associated with the light-cone VQA yields a much faster growth in the approximation ratio during the initial iterations: the average approximation ratio approaches 0.8 at about $p=20$, whereas QAOA requires several hundred iterations to reach a comparable level.

Consistent with Subsection \ref{subsection::5.1}, the lower bound for approximation ratios via Lyapunov functions with either a one or two parameter remains small errors between the continuous procedure and the discrete procedure, once again confirming the reliability of our framework.


\section{Methods}
In this section, we first briefly retrospect combinatorial optimization problems and details of the quantum approximation optimization algorithm. Next, we illustrate intuitions for constructing Lyapunov functions. 

\subsection{Backgrounds}
In optimization problems, our goal is to find a solution such that the objective function $f$ is as large(small) as possible. If $f$ is defined on a ground set $N$, i.e., $f: 2^N \rightarrow \mathbb{R}$, the optimization problem of $\max(\min) f$ is a combinatorial optimization problem. Formally, the combinatorial optimization problem can be expressed as follows:
\begin{equation*}
    \begin{aligned}
         \max \quad & f(S) \\
        \text{s.t.}  \quad &S \in \mathcal{C} \subseteq 2^N
    \end{aligned}
\end{equation*}
where $\mathcal{C}$ is the feasible domain suffered from some constraints. Although some combinatorial optimization problems, such as shortest-path and max-flow, can be solved within polynomial complexity, there are still numerous instances in which it is hard-to-find the optimal solution in acceptable time-consuming, such as Max-Cut, maximum independent set, etc. Consequently, approximation algorithms are substituted for exactly solvable algorithms, which trades off the efficiency and performance. Concretely, the target of approximation algorithms is to design algorithms whose complexity is polynomially bounded with respect to the problem size $n$, such that the solution $S_{ALG}$ yielded by the algorithm adheres some guarantees, i.e., $\frac{f(S_{ALG})}{f(S_{OPT})} \ge r$ with $r \in [0,1]$. It is readily apparent that the approximation ratio serves as the theoretical guarantee of approximation algorithms under the worst cases, given that the approximation algorithm can obtain a solution with $r$-guarantee for all possible instances of this problem. It turns out that the design of approximation algorithms is contingent upon leveraging distinct properties inherent in different problems. Essentially, it entails resolving the following optimization problem.
\begin{equation*}
    \begin{aligned}
         \max \quad & \cfrac{f(S)}{f(S_{OPT})} \\
        \text{s.t.}  \quad & S = ALG(S_0) \\
        & ALG \in \mathcal{ALG}(f,\mathcal{C},\mathcal{P})
    \end{aligned}
\end{equation*}
where $\mathcal{ALG}(f,\mathcal{C},\mathcal{P})$ is the domain of algorithms that crafted based on properties of the function $f$ and are subject to constraints $\mathcal{C}$ of the original problem as well as the complexity $\mathcal{P}$ of the polynomial growth for the problem size $n$. Considering Max-Cut problem, for example, $f$ corresponds to the cut function of general graphs, and $\mathcal{C}$ pertains the power set of the ground set. Semidefinite programming (SDP) is a $0.878$ approximation algorithm with the polynomial complexity $O(poly(n))$ \citep{Goemans1995}, i.e., SDP $\in \mathcal{ALG}(f,\mathcal{C},\mathcal{P})$ and $\frac{f(S_{SDP})}{f(S_{OPT})} \ge 0.878$. Note that, under certain complexity assumptions, such as NP$\neq$P, there is an upper bound for the ratio $\frac{f(S)}{f(S_{OPT})}$ in the context of combinatorial optimization problems within the APX class, i.e., $\sup_{ALG \in \mathcal{ALG}(f,\mathcal{C},\mathcal{P})} \frac{f(S_{ALG})}{f(S_{OPT})} < 1$. For example, under the Unique Games Conjecture, the approximation ratio of SDP is optimal \citep{Khot2007} for Max-Cut, i.e., there is no polynomial approximation algorithm that can achieve the approximation ratio $0.878+\epsilon$ for Max-Cut under the Unique Games Conjecture.   

Alternatively, the combinatorial optimization problem can be converted into a representation with $n$ bits, which corresponds to the ground set $N$ comprising $n$ elements. For each function $f:\{0,1\}^n \rightarrow \mathbb{R}$, we obtain $2^n$ values of $f(\mathbf{x})$ across all possible inputs. Conceivably, this function can be encoded as a Hamiltonian $\mathbf{H}_f$ possessing $2^n$ eigenvalues: 
\begin{equation*}
    \mathbf{H}_f |\mathbf{x} \rangle = f(\mathbf{x}) | \mathbf{x} \rangle \quad \forall \mathbf{x} \in \{0,1\}^n
\end{equation*}
where $\mathbf{H}_f$ is diagonal because the Hamiltonian $\mathbf{H}_f$ corresponding to the function acts as $f$ on all basis. Consequently, for each quantum state $|\psi\rangle = \sum_{\mathbf{x}\in\{0,1\}^n } \alpha_{\mathbf{x}} |\mathbf{x} \rangle$, with $\alpha_{\mathbf{x}} \in \mathbb{C}$ and $\sum_{\mathbf{x}} |\alpha_{\mathbf{x}}|^2 = 1$, the expectation value of the function $f$ with respect to the quantum state $|\psi \rangle$ can be expressed as follows:
\begin{equation*}
    \langle \psi | \mathbf{H}_f  |\psi \rangle = \sum_{\mathbf{x}\in\{0,1\}^n } |\alpha_{\mathbf{x}}|^2 f(\mathbf{x}) = \sum_{S \subseteq N} |\alpha_{S}|^2 f(S)
\end{equation*}
The above map provides an intuition for solving combinatorial optimization problem by employing the quantum computer worked in a $2^n$ dimensional Hilbert space with the above basis.

The QAOA is the seminal framework for generating approximation solutions of combinatorial optimization problems \citep{farhi2014}. This algorithm initiates from the uniform superposition states $|\mathbf{s} \rangle = \frac{1}{\sqrt{2^n}}\sum_{\mathbf{x}} |\mathbf{x}\rangle$. During each round, QAOA successively employs two unitary operators: the function value Hamiltonian operator $U(\mathbf{H}_f,\gamma)$ and the mix Hamiltonian operator $U(\mathbf{H}_x,\beta)$.
\begin{equation*}
    U(\mathbf{H}_f,\gamma)=e^{-i \gamma \mathbf{H}_f}, U(\mathbf{H}_x,\beta) = e^{-i \beta \mathbf{H}_x}, \mathbf{H}_x = \sum_{j=1}^n\mathbf{X}_j
\end{equation*}
where $\gamma \in [0,2\pi]$, $\beta \in [0,\pi]$ and $\mathbf{X}_j$ is the Pauli matrix $\mathbf{X}$ acting on the $j$-th qubit. After $p$ rounds, QAOA returns a quantum state:  
\begin{equation*}
    \begin{aligned}
        |\psi(p)\rangle & = |\bm{\gamma},\bm{\beta} \rangle \\
        & = U(\mathbf{H}_x,\beta_p)U(\mathbf{H}_f,\gamma_p)\cdots U(\mathbf{H}_x,\beta_1)U(\mathbf{H}_f,\gamma_1)|\mathbf{s}\rangle
    \end{aligned}
\end{equation*}
As a result, the main obstacle of QAOA is finding the proper $\bm{\gamma},\bm{\beta}$ to obtain a quantum state such that the expectation of $f$ is as large as possible.  
\begin{equation*}
    \bm{\gamma},\bm{\beta} = \arg\max_{\bm{\gamma},\bm{\beta}} \langle \bm{\gamma},\bm{\beta}| \mathbf{H}_f |\bm{\gamma},\bm{\beta} \rangle
\end{equation*}

Generally, the above optimization problem in classical computers is challenging to obtain the optimal solution $\bm{\gamma}^*,\bm{\beta}^*$. This difficulty arises due to the notoriously barren plateau phenomenon, i.e., the derivative vanishes dramatically, particularly in the case of deep QAOA. 
\subsection{Intuitions from Lyapunov Control}
Recalling the combinatorial optimization problem in its classical version, we can reformulate the combinatorial optimization problem in a quantum version. Given that each subset of the ground set is associated with a $n$-bit representation, in the subsequent discussion, we will refrain from differentiating between these two expressions. 
\begin{equation}
    \begin{aligned}
        \max \quad &\langle \psi | \mathbf{H}_f | \psi \rangle \\
        \text{s.t.} \quad & |\psi\rangle =   \sum_{S \in \mathcal{C}} a_S |S\rangle, \sum_{S \in \mathcal{C}} |a_S|^2 = 1
    \end{aligned}
\end{equation}
where $\mathcal{C} \subseteq 2^N$ is the feasible domain. For designing quantum algorithms, the goal is to maximize approximation ratios of quantum algorithms. Consequently, we can construct the following optimization problem:
\begin{equation}\label{eq:: approximation ratio}
    \begin{aligned}
        & \max \cfrac{\langle \psi(T)| \mathbf{H}_f | \psi(T) \rangle}{\langle \psi^*| \mathbf{H}_f | \psi^* \rangle} \\
        & \text{s.t.}\quad |\psi(T) \rangle \in \mathcal{ALG}(|\psi(0)\rangle,\mathcal{C})
    \end{aligned}
\end{equation}
where $T$ represents the terminal time of the algorithm, and $\mathcal{ALG}$ signifies the domain of quantum algorithms that adhere to some constraints such as complexity requirements and the feasible domain $\mathcal{C}$. The quantum state $|\psi(T)\rangle $, returned by the quantum algorithm after $T$ units of time, constitutes the superposition of feasible states corresponding to all feasible solution. $|\psi ^* \rangle$ is regarded as the optimal solution, i.e., $|\psi ^* \rangle = \arg\max_{|\psi\rangle} \langle \psi | \mathbf{H}_f | \psi \rangle$. Concretely, if we denote $OPT = \{S \in \mathcal{C}| f(S) \ge f(V), \forall V \in \mathcal{C}\}$, then $|\psi^* \rangle = \sum_{S \in OPT} a_S |S\rangle$. Essentially, our target is to devise a quantum algorithm that evolves the state $| \psi(0) \rangle$ at time $0$ to the state $|\psi(T) \rangle$ at time $T$ such that the expectation of the objective function $f$ is as large as possible. In the context of quantum approximation algorithms, it will be strengthen to ensure that the expectation of the objective function $f$ under the worst instance is as large as possible. Considering the Schr\"{o}dinger evolution equation:
\begin{equation*}
       i\cfrac{d |\psi(t)\rangle}{ dt } = \mathbf{H}(t)  |\psi(t)\rangle \Rightarrow  |\psi(t)\rangle = \mathcal{T} e^{-i \int_0^t \mathbf{H}(s) ds }  |\psi(0)\rangle 
\end{equation*}
where $\mathcal{T}$ is the time-ordering operator and the Planck's constant has been ignored. Consequently, the above optimization problem (Eq. (\ref{eq:: approximation ratio})) can be rewritten as follows: 
\begin{equation}\label{eq: approximation ratio quantum version}
    \begin{aligned}
        & \max \cfrac{\langle \psi(T)| \mathbf{H}_f | \psi(T) \rangle}{\langle \psi^*| \mathbf{H}_f | \psi^* \rangle} \\
         \text{s.t.}\quad &|\psi(T)\rangle = \mathcal{T} e^{-i \int_0^T \mathbf{H}(s) ds }  |\psi(0)\rangle \\
        &  \mathbf{H}(s) \text{ is Hermitian operator}
    \end{aligned}
\end{equation}
Theoretically, if our intention is to design a quantum approximation algorithm with some theoretical guarantees, the essential obstacle is to design a proper Hamiltonian $\mathbf{H}(t)$. Denote the optimal approximation ratio and the optimal Hamiltonian as follows:
\begin{equation*}
    \begin{aligned}
        \lambda^* & = \max \cfrac{\langle \psi(T)| \mathbf{H}_f | \psi(T) \rangle}{\langle \psi^*| \mathbf{H}_f | \psi^* \rangle}\\
        \mathbf{H}^*(t) &= \arg\max \cfrac{\langle \psi(T)| \mathbf{H}_f | \psi(T) \rangle}{\langle \psi^*| \mathbf{H}_f | \psi^* \rangle}
    \end{aligned}
\end{equation*}
We first introduce the following lemma that facilitates our reconstruction of the ratio maximization problem into the difference maximization problem.
\begin{lemma}\label{lm::ratio difference}
    (1) If $\lambda^*$ is the optimal ratio, then the following optimization problem has the same optimal solution:
    \begin{equation}
        \begin{aligned}
            \mathbf{H}_1^*(t) &= \arg\max_{\mathbf{H}(t)} \cfrac{\langle \psi(T)| \mathbf{H}_f | \psi(T) \rangle}{\langle \psi^*| \mathbf{H}_f | \psi^* \rangle} \\
            &= \arg\max_{\mathbf{H}(t) } \langle \psi(T)| \mathbf{H}_f | \psi(T) \rangle- \lambda^* \langle \psi^*| \mathbf{H}_f | \psi^* \rangle
        \end{aligned}
    \end{equation}
    (2) If $c^*$ is the optimal difference, i.e., $c^*=\max_{\mathbf{H}(t) } \langle \psi(T)| \mathbf{H}_f | \psi(T) \rangle- \langle \psi^*| \mathbf{H}_f | \psi^* \rangle$, then the following optimization problem has the same optimal solution:
    \begin{equation}
        \begin{aligned}
            \mathbf{H}_2^*(t) &= \arg\max_{\mathbf{H}(t) } \langle \psi(T)| \mathbf{H}_f | \psi(T) \rangle- \langle \psi^*| \mathbf{H}_f | \psi^* \rangle \\
            &= \arg\max_{\mathbf{H}(t)} \cfrac{\langle \psi(T)| \mathbf{H}_f | \psi(T) \rangle}{\langle \psi^*| \mathbf{H}_f | \psi^*\rangle +c^*}
        \end{aligned}
    \end{equation}
\end{lemma}
The detailed proof of Lemma \ref{lm::ratio difference} can be found in Appendix \ref{Proof of Threom thm::ratio difference}. The above lemma demonstrates an equivalent transformation from the ratio maximization problem to a difference problem. Inspired by Lemma \ref{lm::ratio difference}, we can reformulate the Eq.(\ref{eq:: approximation ratio}) as follows:
\begin{equation}\label{eq:: approximation ratio quantum version b}
    \begin{aligned}
         \max & \quad \langle \psi(T)| \mathbf{H}_f | \psi(T) \rangle - \lambda^* \langle \psi^*| \mathbf{H}_f | \psi^* \rangle \\
        \text{s.t.} & \quad |\psi(T) \rangle \in \mathcal{ALG}(|\psi(0)\rangle,\mathcal{C})
    \end{aligned}
\end{equation}
Generally, the optimization problem described in Eq.(\ref{eq:: approximation ratio quantum version b}) is hard-to-solve exactly. Therefore, we resort to employ indirect approaches to derive an solution. Regarding the execute time $t$ as a parameter, it implies that the approximation ratio can be expressed as a function of time $t$ within the context of approximation algorithms. Consequently, we can reformulate the objective function presented in Eq.(\ref{eq:: approximation ratio quantum version b}) as follows:
\begin{equation*}
    E(t) = \langle \psi(t)| \mathbf{H}_f | \psi(t) \rangle - \lambda(t) \langle \psi^*| \mathbf{H}_f | \psi^* \rangle
\end{equation*}
If we can control the above equation properly, we can indirectly obtain an appropriate Hamiltonian $\mathbf{H}(t)$ such that the quantum approximation algorithm based on $\mathbf{H}(t)$ possesses a provably reliable approximation ratio.

\section{Discussion}

In this work, we propose a universal framework for quantum approximation algorithms designing and analyzing by constructing time-dependent Lyapunov functions with one and two parameters, respectively. This framework effectively exploits the information generated during the execution of the algorithm and enables a dynamic, adaptive analysis of the incremental improvement of the approximation ratio at each iteration. Through numerical experiments, we not only validate the reliability of the proposed lower bound for the true approximation ratio but also provide new evidences that adaptive quantum algorithms can possibly outperform the best-known classical algorithms in terms of performance. This work partially fills the gap in the theoretical performance analysis of quantum approximation algorithms. Nonetheless, the proposed framework has certain limitations, and many aspects remain open for further investigation.

\textit{a. Parameter Selection.} Part of the parameter selection in our algorithm is inspired by the measurement-feedback concept, which alleviates some of the burden associated with parameter tuning while ensuring that the algorithm evolves along the direction of increasing observable expectation. However, in our numerical experiments, due to the large circuit depth, the parameter $\beta(t)$ was chosen from a few simple candidate functions based on empirical performance. How to determine an optimal sequence of coefficients $\beta(t)$ for deep circuits remains an open and important problem.


\textit{b. Selection of the Upper Bound for the Solution.} We have verified the reliability of the proposed framework for the lower bound of the approximation ratio from both theoretical analysis and numerical experiments. However, in the numerical experiments on the Max-Cut problem, we adopted a simple upper bound, namely the total number of edges $|E|$ of the undirected graph, which is independent of the iterative process. This simplification led to a loose lower bound of the approximation ratio. Therefore, in practical applications of the framework, if specific methods can be developed to make full use of the information generated during the execution of the algorithm, either through analytical approaches or heuristic estimation of the optimal solution's upper bound, the overall performance of the framework could be significantly improved.

\section*{Acknowledgements}
This work was supported in part by the Quantum Science and Technology-National Science and Technology Major Project under Grant No. 2024ZD0300500 and the National Natural Science Foundation of China Grants No. 62325210, 62272441, 12501450, 12571382.

\bibliographystyle{plain}
\bibliography{ref}

\onecolumngrid
\newpage


\appendix

\section{Discrete Procedures for Lyapunov Function with One Parameter} \label{Appendix::One-Parameter-Discretion}

In this appendix, we present the discrete time formulation of the approximation ratio for the algorithm described by Eq. (\ref{eq::single_para_estimate}) under the one parameter Lyapunov framework, which facilitates the direct design of discrete-time quantum algorithms. Firstly, we define the discrete Lyapunov function, also known as the potential function. 
\begin{definition}[Potential Function]
    The discrete Lyapunov function at time $t_j$ is as follows:
    \begin{equation}
        E(t_j) = \langle \psi(t_j)| \mathbf{H}_f | \psi(t_j) \rangle- \lambda(t_j) \langle \psi^*| \mathbf{H}_f | \psi^* \rangle
    \end{equation}
    where $\lambda(t_j)$ is a non-negative non-decreasing sequence. The upper bound of the optimal solution at time $t_j$ is as follows:
    \begin{equation}
        0 \le \langle \psi^*| \mathbf{H}_f | \psi^* \rangle \le \langle \psi(t_j)| \mathbf{Q}(t_j) | \psi(t_j) \rangle
    \end{equation}
    Note that the upper bound of optimal solution can be easily estimated prior to the design of the algorithm, which depends on the mathematical property of the objective function.
\end{definition}
Since the potential function is discrete, it essentially constitutes a sequence that is only defined at discrete time points $t_j$. In continuous times, the approximation algorithm design hinges on the monotonicity of Lyapunov function. Analogously, in discrete times, we also need to construct a non-decreasing potential function. Then, the increment of the potential function from time $t_j$ to time $t_{j+1}$ can be expressed as follows:
\begin{equation}\label{eq::potential increment}
    \begin{aligned}
        E(t_{j+1}) - E(t_j) & = \langle \psi(t_{j+1})| \mathbf{H}_f | \psi(t_{j+1}) \rangle- \lambda(t_{j+1}) \langle \psi^*| \mathbf{H}_f | \psi^* \rangle  - \langle \psi(t_j)| \mathbf{H}_f | \psi(t_j) \rangle + \lambda(t_j) \langle \psi^*| \mathbf{H}_f | \psi^* \rangle \\
        & \ge \langle \psi(t_{j+1})| \mathbf{H}_f | \psi(t_{j+1}) \rangle - \langle \psi(t_j)| \mathbf{H}_f | \psi(t_j) \rangle - \left[ \lambda(t_{j+1}) - \lambda(t_j) \right] \langle \psi(t_j)| \mathbf{Q}(t_j) | \psi(t_j) \rangle \\
        & \ge -Err(t_j)
    \end{aligned}
\end{equation}
where $Err(t_j)$ is the error term, i.e., $Err(t_j)$ signifies an acceptable extent to which monotonicity can be disrupted. If the total round number is $N$ and the end time is $T$, then the approximation ratio of algorithms based on the above potential function is given as follows:
\begin{equation}\label{eq::approximation ratio c}
    \langle \psi(T)| \mathbf{H}_f | \psi(T) \rangle \ge (\lambda(T)-\lambda(0)) \langle \psi^*| \mathbf{H}_f | \psi^* \rangle + \langle \psi(0)| \mathbf{H}_f | \psi(0) \rangle -  \sum_{j=0}^{N-1} Err(t_j)
\end{equation}
Regarding $Err(t_j)$, there are primarily two resources: (1) errors that stem from the transition continuous times to discrete times; and (2) the potential sacrifice of monotonicity in pursuit of a larger $\lambda(t_{j+1})-\lambda(t_j)$. Therefore, our target is to choose Hamiltonian $\mathbf{H}(t_j)$ and $\lambda(t_j)$ at each discrete time $t_j$ such that the difference $\lambda(T) - \lambda(0)$ is as large as possible while keeping the absolute value of errors as small as possible. Intuitively, we start from Eq.(\ref{eq::potential increment}) to determine the Hamiltonian $\mathbf{H}(t_j)$ and $\lambda(t_j)$ at each time $t_j$. Before describing details, we introduce three lemmas that serve as fundamental ingredients in the analysis of errors. In this work, we consider the Hamiltonian is Hermitian. 
\begin{lemma}\label{lemma: two time-independent}
    \begin{equation}
        |\langle \psi(t)|\mathbf{A} \mathbf{H}_f  \mathbf{A}^{\dagger}| \psi(t) \rangle - \langle \psi(t)|\mathbf{B} \mathbf{H}_f  \mathbf{B}^{\dagger}| \psi(t) \rangle| \le \left\| \mathbf{A} - \mathbf{B}\right \| \left\| \mathbf{H}_f \right\| (\left\| \mathbf{A} \right\| + \left\| \mathbf{B} \right\|)
    \end{equation}
\end{lemma}
\begin{proof}
    \begin{equation*}
        \begin{aligned}
            & \left\|\langle \psi(t)|\mathbf{A} \mathbf{H}_f  \mathbf{A}^{\dagger}| \psi(t) \rangle - \langle \psi(t)|\mathbf{B} \mathbf{H}_f  \mathbf{B}^{\dagger}| \psi(t) \rangle\right\| \\
            = & \left\|\langle \psi(t)|\mathbf{A} \mathbf{H}_f  \mathbf{A}^{\dagger}| \psi(t) \rangle - \langle \psi(t)|\mathbf{B} \mathbf{H}_f  \mathbf{A}^{\dagger}| \psi(t) \rangle + \langle \psi(t)|\mathbf{B} \mathbf{H}_f  \mathbf{A}^{\dagger}| \psi(t) \rangle - \langle \psi(t)|\mathbf{B} \mathbf{H}_f  \mathbf{B}^{\dagger}| \psi(t) \rangle\right\| \\
            \le & \left\| \mathbf{A} - \mathbf{B}\right \| \left\| \mathbf{H}_f \right\| (\left\| \mathbf{A} \right\| + \left\| \mathbf{B} \right\|) 
        \end{aligned}
    \end{equation*}
\end{proof}

\begin{lemma}[A Rephrased Version of Lemma 1 in \citep{Berry2020}, Appendix B in \citep{Tran2019}]\label{lemma: two time-dependent}
    If we have two Hamiltonian $\mathbf{H}_1(t)$ and $\mathbf{H}_2(t)$, then
    \begin{equation}
        \left\| \mathcal{T} e^{-i \int_0^T \mathbf{H}_1(s) ds } - \mathcal{T} e^{-i \int_0^T \mathbf{H}_2(s) ds } \right\| \le \int_0^{T} \left\| \mathbf{H}_1(s) - \mathbf{H}_2(s) \right\| ds
    \end{equation}
\end{lemma}

\begin{proof}
    Firstly, we have 
    \begin{equation*}
        \begin{array}{c}
            \cfrac{d |\psi(t)\rangle}{d t} = \cfrac{d \mathcal{T} e^{-i \int_0^t \mathbf{H}(s) ds }  |\psi(0)\rangle}{d t} = -i \mathbf{H}(t) \mathcal{T} e^{-i \int_0^t \mathbf{H}(s) ds }  |\psi(0)\rangle \\
            \cfrac{d \langle \psi(t)|}{d t} = \cfrac{d \langle \psi(0) |\mathcal{T} e^{i \int_0^t \mathbf{H}(s) ds }  }{d t} =\langle\psi(0)| i \mathcal{T} e^{i \int_0^t \mathbf{H}(s) ds }  \mathbf{H}(t)
        \end{array}
    \end{equation*}
    In addition, $ \left\| \mathcal{T} e^{i \int_0^T \mathbf{H}(s) ds } \right\| = 1$ because it is unitary, we have 
    \begin{equation*}
        \begin{aligned}
            & \left\| \mathcal{T} e^{-i \int_0^T \mathbf{H}_1(s) ds } - \mathcal{T} e^{-i \int_0^T \mathbf{H}_2(s) ds } \right\|  \le \left\| \mathcal{T} e^{i \int_0^T \mathbf{H}_2(s) ds }\mathcal{T} e^{-i \int_0^T \mathbf{H}_1(s) ds} -\mathbf{I}\right \| = \left \| \int_0^T dl \cfrac{d \mathcal{T} e^{i \int_0^l \mathbf{H}_2(s) ds} \mathcal{T} e^{-i \int_0^l \mathbf{H}_1(s) ds } } {d l} \right\| \\
            & = \left\| \int_0^T dl \mathcal{T} e^{i \int_0^l \mathbf{H}_2(s) ds }  \mathbf{H}_2(l)  \mathcal{T} e^{-i \int_0^l \mathbf{H}_1(s) ds } - \mathcal{T} e^{i \int_0^l \mathbf{H}_2(s) ds }  \mathbf{H}_1(l)  \mathcal{T} e^{-i \int_0^l \mathbf{H}_1(s) ds } \right\| \\
            & = \left\| \int_0^T dl \mathcal{T} e^{i \int_0^l \mathbf{H}_2(s) ds } ( \mathbf{H}_2(l) -\mathbf{H}_1(l) ) \mathcal{T} e^{-i \int_0^l \mathbf{H}_1(s) ds } \right\| \le  \int_0^T dl \left\| \mathcal{T} e^{i \int_0^l \mathbf{H}_2(s) ds } ( \mathbf{H}_2(l) -\mathbf{H}_1(l) ) \mathcal{T} e^{-i \int_0^l \mathbf{H}_1(s) ds } \right\| \\
            & \le \int_0^T dl \left\| \mathbf{H}_2(l) -\mathbf{H}_1(l) \right\|
        \end{aligned}
    \end{equation*}
\end{proof}

\begin{lemma}
    \begin{equation*}
        \begin{aligned}
            & \left \| e^{i \mathbf{H}^{f}_1(t_j) (t_{j+1}-t_j)} \cdots e^{i \mathbf{H}^{f}_{n_1}(t_j) (t_{j+1}-t_j)} e^{i \mathbf{H}^{\hat{f}}_1(t_j) (t_{j+1}-t_j)} \cdots e^{i \mathbf{H}^{\hat{f}}_{n_2}(t_j) (t_{j+1}-t_j)} -  e^{i \mathbf{H}(t_j) (t_{j+1}-t_j)} \right \| \\
            & \le \sum_{k=1}^{n_1} \left\| \mathbf{H}^f_k\right\| (t_{j+1}-t_j) + \sum_{k=1}^{n_2-1} \left\| \mathbf{H}_k^{\hat{f}}(t_j) \right\| (t_{j+1}-t_j) + \left\| \mathbf{H}^{\hat{f}}_{n_2}(t_j)- \mathbf{H} (t_j)\right\| (t_{j+1}-t_j) 
        \end{aligned}
    \end{equation*}
\end{lemma}
\begin{proof}
Firstly, we can rewrite the difference by the accumulation.
\begin{equation*}
    \begin{aligned}
        & e^{i \mathbf{H}^{f}_1(t_j) (t_{j+1}-t_j)} \cdots e^{i \mathbf{H}^{f}_{n_1}(t_j) (t_{j+1}-t_j)} e^{i \mathbf{H}^{\hat{f}}_1(t_j) (t_{j+1}-t_j)} \cdots e^{i \mathbf{H}^{\hat{f}}_{n_2}(t_j) (t_{j+1}-t_j)} -  e^{i \mathbf{H}(t_j) (t_{j+1}-t_j)} \\
        & = e^{i \mathbf{H}^{f}_1(t_j) (t_{j+1}-t_j)} \cdots e^{i \mathbf{H}^{f}_{n_1}(t_j) (t_{j+1}-t_j)} e^{i \mathbf{H}^{\hat{f}}_1(t_j) (t_{j+1}-t_j)} \cdots e^{i \mathbf{H}^{\hat{f}}_{n_2}(t_j) (t_{j+1}-t_j)} \\
        & \quad  - e^{i \mathbf{H}^{f}_2(t_j) (t_{j+1}-t_j)} \cdots e^{i \mathbf{H}^{f}_{n_1}(t_j) (t_{j+1}-t_j)} e^{i \mathbf{H}^{\hat{f}}_1(t_j) (t_{j+1}-t_j)} \cdots e^{i \mathbf{H}^{\hat{f}}_{n_2}(t_j) (t_{j+1}-t_j)} \\
        & \quad + e^{i \mathbf{H}^{f}_2(t_j) (t_{j+1}-t_j)} \cdots e^{i \mathbf{H}^{f}_{n_1}(t_j) (t_{j+1}-t_j)} e^{i \mathbf{H}^{\hat{f}}_1(t_j) (t_{j+1}-t_j)} \cdots e^{i \mathbf{H}^{\hat{f}}_{n_2}(t_j) (t_{j+1}-t_j)} \\
        & \quad - e^{i \mathbf{H}^{f}_3(t_j) (t_{j+1}-t_j)} \cdots e^{i \mathbf{H}^{f}_{n_1}(t_j) (t_{j+1}-t_j)} e^{i \mathbf{H}^{\hat{f}}_1(t_j) (t_{j+1}-t_j)} \cdots e^{i \mathbf{H}^{\hat{f}}_{n_2}(t_j) (t_{j+1}-t_j)} \\
        & \quad \vdots \\
        & \quad +e^{i \mathbf{H}^{\hat{f}}_{n_2}(t_j) (t_{j+1}-t_j)} - e^{i \mathbf{H}(t_j) (t_{j+1}-t_j)}  \\
    \end{aligned}
\end{equation*}
For convenience, we denote $\mathbf{H}_p(t_j) = \mathbf{H}^{f}_p(t_j)$ if $p \le n_1$ and $\mathbf{H}_p(t_j) = \mathbf{H}^{\hat{f}}_{p-n_1}(t_j)$ otherwise. Adopting the triangle inequality, we have:
\begin{equation*}
    \begin{aligned}
        & \left \| e^{i \mathbf{H}^{f}_1(t_j) (t_{j+1}-t_j)} \cdots e^{i \mathbf{H}^{f}_{n_1}(t_j) (t_{j+1}-t_j)} e^{i \mathbf{H}^{\hat{f}}_1(t_j) (t_{j+1}-t_j)} \cdots e^{i \mathbf{H}^{\hat{f}}_{n_2}(t_j) (t_{j+1}-t_j)} -  e^{i \mathbf{H}(t_j) (t_{j+1}-t_j)} \right \| \\
        \le & \sum_{k=1}^{n_1+n_2-1} \left\| e^{i \prod_{p=k}^{n_1+n_2} \mathbf{H}_p(t_j)(t_{j+1}-t_j)} - e^{i \prod_{p=k+1}^{n_1+n_2} \mathbf{H}_p(t_j) (t_{j+1}-t_j)}\right\| + \left\| e^{i \mathbf{H}^{\hat{f}}_{n_2}(t_j) (t_{j+1}-t_j)} - e^{i \mathbf{H}(t_j) (t_{j+1}-t_j)} \right\| 
    \end{aligned}
\end{equation*}
Leveraging Lemma \ref{lemma: two time-dependent}, we have:
\begin{equation*}
    \left\| e^{i \prod_{p=k}^{n_1+n_2} \mathbf{H}_p(t_j)(t_{j+1}-t_j)} - e^{i \prod_{p=k+1}^{n_1+n_2} \mathbf{H}_p(t_j) (t_{j+1}-t_j)}\right\| \le \left \| e^{i \mathbf{H}_k(t_j) (t_{j+1}-t_j)} - \mathbf{I} \right\| \le \int_{t_j}^{t_{j+1}} dl \left\| \mathbf{H}_k(t_j)\right\|
\end{equation*}
Combining the above inequalities, we have:
\begin{equation*}
    \begin{aligned}
        & \left \| e^{i \mathbf{H}^{f}_1(t_j) (t_{j+1}-t_j)} \cdots e^{i \mathbf{H}^{f}_{n_1}(t_j) (t_{j+1}-t_j)} e^{i \mathbf{H}^{\hat{f}}_1(t_j) (t_{j+1}-t_j)} \cdots e^{i \mathbf{H}^{\hat{f}}_{n_2}(t_j) (t_{j+1}-t_j)} -  e^{i \mathbf{H}(t_j) (t_{j+1}-t_j)} \right \| \\
        \le & \sum_{k=1}^{n_1+n_2-1} \int_{t_j}^{t_{j+1}} dl \left\| \mathbf{H}_k(t_j)\right\| + \int_{t_j}^{t_{j+1}} dl \left\| \mathbf{H}_{n_1+n_2}(t_j) - \mathbf{H}(t_j)\right\| \\
        = &\sum_{k=1}^{n_1} \left\| \mathbf{H}^f_k(t_j)\right\| (t_{j+1}-t_j) + \sum_{k=1}^{n_2-1} \left\| \mathbf{H}_k^{\hat{f}}(t_j) \right\| (t_{j+1}-t_j) + \left\| \mathbf{H}^{\hat{f}}_{n_2}(t_j)- \mathbf{H} (t_j)\right\| (t_{j+1}-t_j)
    \end{aligned}
\end{equation*}
\end{proof}

We impose certain restrictions on our algorithm. During each period time interval $t_{j+1}-t_j$, the Hamiltonian remains time-independent $\mathbf{H}(t_j)$ rather than $\mathbf{H}(t)$ for $ \forall t \in [t_j,t_{j+1})$. Essentially, our algorithm evolves under a time-independent Hamiltonian within each discrete time period. In addition, we consider the implementation of our algorithm in quantum circuits. Here, we consider the linear combination of some efficiently implemented Hamiltonian. Consequently, we stipulate that our algorithm can be efficiently implemented in quantum circuits through two components $e^{\mathbf{H}^{f}_1(t_j)},\cdots ,e^{\mathbf{H}^{f}_{n_1}(t_j)}$ and $e^{\mathbf{H}^{\hat{f}}_1(t_j)},\cdots ,e^{\mathbf{H}^{\hat{f}}_{n_2}(t_j)}$, where $\mathbf{H}^{f}_k(t_j)$ commutes with $\mathbf{H}_f$, while $\mathbf{H}^{\hat{f}}_k(t_j)$ does not commute with $\mathbf{H}_f$. For notational convenience, we define $\mathbf{H}(t_j) = \mathbf{H}^f(t_j) + \mathbf{H}^{\hat{f}}(t_j) = \sum_{k=1}^{n_1} \mathbf{H}^{f}_k (t_j) + \sum_{k=1}^{n_2} \mathbf{H}^{\hat{f}}_k (t_j)$. The increment of the potential function from time $t_j$ to time $t_{j+1}$ can be expressed as follows: 
\begin{equation}
    \begin{aligned}
        & E(t_{j+1}) - E(t_j) = \langle \psi(t_{j+1})| \mathbf{H}_f | \psi(t_{j+1}) \rangle - \langle \psi(t_j)| \mathbf{H}_f | \psi(t_j) \rangle  - \left[ \lambda(t_{j+1}) - \lambda(t_j) \right] \langle \psi^*| \mathbf{H}_f | \psi^* \rangle  \\
        & = \langle \psi(t_j)|e^{i \mathbf{H}^f_1(t_j)(t_{j+1}-t_j)}\cdots e^{i \mathbf{H}^f_{n_1}(t_j)(t_{j+1}-t_j)} e^{i \mathbf{H}^{\hat{f}}_1(t_j)(t_{j+1}-t_j)}\cdots e^{i \mathbf{H}^{\hat{f}}_{n_2}(t_j)(t_{j+1}-t_j)} \\
        & \qquad \qquad \quad \mathbf{H}_f e^{-i \mathbf{H}^{\hat{f}}_{n_2}(t_j)(t_{j+1}-t_j)}\cdots e^{-i \mathbf{H}^{\hat{f}}_1 (t_j)(t_{j+1}-t_j)} e^{-i \mathbf{H}^f_{n_1}(t_j)(t_{j+1}-t_j)}\cdots e^{-i \mathbf{H}^f_1(t_j)(t_{j+1}-t_j)}   | \psi(t_j)\rangle \\
        & \quad - \langle \psi(t_j)| \mathbf{H}_f | \psi(t_j) \rangle - \left[ \lambda(t_{j+1}) - \lambda(t_j) \right] \langle \psi^*| \mathbf{H}_f | \psi^* \rangle \\
        & \ge \langle \psi(t_j)|e^{i \mathbf{H}(t_j)(t_{j+1}-t_j)} \mathbf{H}_f e^{-i \mathbf{H}(t_j)(t_{j+1}-t_j)} | \psi(t_j)\rangle - \langle \psi(t_j)| \mathbf{H}_f | \psi(t_j) \rangle - \left[ \lambda(t_{j+1}) - \lambda(t_j) \right] \langle \psi(t_j)| \mathbf{Q}(t_j) | \psi(t_j) \rangle \\
        & \quad -2 \left\| \mathbf{H}_f \right \| \left( \sum_{k=1}^{n_1} \left\| \mathbf{H}^f_{k}(t_j) \right\| + \sum_{k=1}^{n_2-1} \left\| \mathbf{H}_k^{\hat{f}}(t_j)  \right\|  + \left\| \sum_{k=1}^{n_1} \mathbf{H}^f_k(t_j) + \sum_{k=1}^{n_2-1} \mathbf{H}_k^{\hat{f}} (t_j)\right\|  \right) (t_{j+1}-t_j) \ge -\epsilon
    \end{aligned}
\end{equation}
The first inequality comes from the above lemmas. Rearranging the above inequality, the approximation ratio $\lambda$ has the following explicit formula:
\begin{equation}
    \begin{aligned}
        & \cfrac{\langle \psi(t_j)|e^{i \mathbf{H}(t_j)(t_{j+1}-t_j)} \mathbf{H}_f e^{-i \mathbf{H}(t_j)(t_{j+1}-t_j)} | \psi(t_j)\rangle - \langle \psi(t_j)| \mathbf{H}_f | \psi(t_j) \rangle + \epsilon}{\langle \psi(t_j)| \mathbf{Q}(t_j) | \psi(t_j) \rangle} \\
        & - \cfrac{2 \left\| \mathbf{H}_f \right \| \left( \sum_{k=1}^{n_1} \left\| \mathbf{H}^f_{k}(t_j) \right\| + \sum_{k=1}^{n_2-1} \left\| \mathbf{H}_k^{\hat{f}}(t_j)  \right\|  + \left\| \sum_{k=1}^{n_1} \mathbf{H}^f_k(t_j) + \sum_{k=1}^{n_2-1} \mathbf{H}_k^{\hat{f}} (t_j)\right\|  \right) (t_{j+1}-t_j)}{\langle \psi(t_j)| \mathbf{Q}(t_j) | \psi(t_j) \rangle} \\
        & \ge  \lambda(t_{j+1} ) - \lambda(t_j)
    \end{aligned}
\end{equation}

Leveraging Eq.(\ref{eq::potential increment}), our objective is to maximize $\lambda(T)-\lambda(0)$ and minimize $\sum_{j=0}^{N-1} Err(t_j)$ simultaneously. The pivotal task becomes identifying a proper $\mathbf{H}(t_j)$ such that the above inequality is as large as possible while $\epsilon$ is as small as possible. To achieve this goal, we employ the Taylor expansion technique. There are two primary reasons for this choice. Firstly, the objective function at the current state $| \psi(t_j)\rangle$ is known before choosing $\mathbf{H}(t_j)$. By applying the Taylor expansion, we can effectively eliminate it. Secondly, we can employ the Taylor expansion to describe the future, as yet unknown state $|\psi(t_{j+1})\rangle$ in terms of the presently known state $|\psi(t_{j})\rangle$. This allows us to more effectively manipulate and optimize the expression with respect to our target.
\begin{equation}\label{eq: taylor expansion}
    \begin{aligned}
        & \langle \psi(t_j)|e^{i \mathbf{H}(t_j)(t_{j+1}-t_j)}  \mathbf{H}_f  e^{-i \mathbf{H}(t_j) (t_{j+1}-t_j)} | \psi(t_j)\rangle - \langle \psi(t_j)| \mathbf{H}_f | \psi(t_j) \rangle \\
        & = \langle \psi(t_j)| i [ \mathbf{H}^{\hat{f}}(t_j), \mathbf{H}_f] | \psi(t_j) \rangle (t_{j+1}-t_j) + \frac{1}{2!}\langle \psi(t_j)| [\mathbf{H}(t_j),[\mathbf{H}_f,\mathbf{H}^{\hat{f}}(t_j)]]  | \psi(t_j) \rangle  (t_{j+1}-t_j)^2 \\
        & \quad + \frac{1}{3!} \langle \psi(t_j)| [\mathbf{H}(t_j),[\mathbf{H}(t_j),[\mathbf{H}_f,\mathbf{H}^{\hat{f}}(t_j)]]]  | \psi(t_j) \rangle (t_{j+1}-t_j)^3 + \cdots  
    \end{aligned}
\end{equation}

Based on the above equation, we construct a solution to approach the approximation ratio. Intuitively, our approach entails choosing a proper $\mathbf{H}(t_j)$ such that $\langle \psi(t_j)| i [ \mathbf{H}^{\hat{f}}(t_j), \mathbf{H}_f] | \psi(t_j) \rangle (t_{j+1}-t_j)$ is non-negative, while simultaneously keeping the absolute value of the remainder term as small as possible. Similarly, it should be emphasized that the above control procedure only manipulates the non-commutative component, as the commutative component does not directly contribute to the objective function. We choose for the following:
\begin{equation}\label{eq: Lyapunov 1 hamiltonian discrete}
    \mathbf{H}^{f}(t_j) = \sum_{k=1}^{n_1} \eta_k(t_j)\mathbf{H}^{f}_k, \quad \mathbf{H}^{\hat{f}}(t_j) = \sum_{k=1}^{n_2} \alpha_k(t_j) \mathbf{H}^{\hat{f}}_k
\end{equation}
where $\mathbf{H}^{f}_k$ and $\mathbf{H}^{\hat{f}}_k$ can be effectively implemented on quantum circuits. Using the feedback-control and denoting $O_k(t_j) = \langle \psi(t_j)| i [\mathbf{H}^{\hat{f}}_k, \mathbf{H}_f] | \psi(t_j) \rangle$, we choose $\alpha_k(t_j) = \beta_k(t_j)O_k(t_j)$. If $\beta_k(t_j) \ge 0$, we have: 
\begin{equation}\label{eq: approximation ratio lyapunov 1}
    \begin{aligned}
        \cfrac{\langle \psi(t_j)| i [ \mathbf{H}^{\hat{f}}(t_j), \mathbf{H}_f] | \psi(t_j) \rangle (t_{j+1}-t_j)}{\langle \psi(t_j)| \mathbf{Q}(t_j) | \psi(t_j) \rangle} = \cfrac{\sum_{k=1}^{n_2} \beta_k(t_j)O_k^2(t_j) (t_{j+1}-t_j)}{\langle \psi(t_j)| \mathbf{Q}(t_j) | \psi(t_j) \rangle} \ge 0
    \end{aligned}
\end{equation}
As a result, the difference of the approximation ratio at each round can be bounded from below by the following inequality:
\begin{equation}
    \begin{array}{c}
        \lambda(t_{j+1} ) - \lambda(t_j)  =\cfrac{\Delta^1 (t_{j+1}-t_j)}{\langle \psi(t_j)| \mathbf{Q}(t_j) | \psi(t_j) \rangle} \\
        \begin{aligned}
            \epsilon \ge &\cfrac{\sum_{k=2}^{\infty} |\Delta^k(t_{j+1}-t_j)|}{\langle \psi(t_j)| \mathbf{Q}(t_j) | \psi(t_j) \rangle} \\
            & +\cfrac{2 \left\| \mathbf{H}_f \right \| \left( \sum_{k=1}^{n_1} \left\| \mathbf{H}^f_{k}(t_j)  \right\| + \sum_{k=1}^{n_2-1} \left\| \mathbf{H}_k^{\hat{f}}(t_j) \right\|  + \left\| \sum_{k=1}^{n_1} \mathbf{H}^f_k(t_j) + \sum_{k=1}^{n_2-1} \mathbf{H}_k^{\hat{f}} (t_j)\right\|  \right) (t_{j+1}-t_j)}{\langle \psi(t_j)| \mathbf{Q}(t_j) | \psi(t_j) \rangle}
        \end{aligned}
    \end{array}
\end{equation}
where $\Delta^k (t_{j+1}-t_j)$ is denoted as $k$-th term in Taylor expansion. The above formula implies that our aspiration is to make the first term in Taylor expansion as large as possible and the reminder terms in Taylor expansion are errors as small as possible. Next, we analyze the bound of the error.
\begin{lemma}
    \begin{equation*}
        \sum_{k=2}^{\infty} |\Delta^k(t_{j+1}-t_j)| \le \cfrac{ 2 \left\| \mathbf{H}_f\right\|\left\| \mathbf{H}^{\hat{f}}(t_j)\right\|\left \| \mathbf{H}(t_j) \right\| (t_{j+1} - t_j)^2}{(1-\left \| \mathbf{H}(t_j) \right\| (t_{j+1} - t_j))}
    \end{equation*}
\end{lemma}
\begin{proof}
    It is easy to check $|\Delta^k(t_{j+1}-t_j)|$ has the following explicit formulation:
    \begin{equation*}
        \begin{aligned}
            | \Delta^k (t_{j+1}-t_j) | & = \left| \frac{1}{k!} \langle \psi(t_j)|[\mathbf{H}(t_j), \cdots [\mathbf{H}(t_j),[\mathbf{H}(t_j),[\mathbf{H}_f,\mathbf{H}(t_j)]]]\cdots]  | \psi(t_j) \rangle \right| (t_{j+1}-t_j)^k \\
            & \le \cfrac{2^k \left \| \mathbf{H}(t_j) \right\|^{k-1} \left\| \mathbf{H}^{\hat{f}}(t_j) \right\| \left\| \mathbf{H}_f\right\|  (t_{j+1}-t_j)^k}{k!} \\
            & \le 2 \left \| \mathbf{H}(t_j) \right\|^{k-1} \left\| \mathbf{H}^{\hat{f}}(t_j) \right\| \left\| \mathbf{H}_f\right\|  (t_{j+1}-t_j)^k
        \end{aligned}
    \end{equation*}
    Therefore, we have:
    \begin{equation*}
        \begin{aligned}
            \sum_{k=2}^{\infty} |\Delta^k(t_{j+1}-t_j)| & \le \sum_{k=2}^\infty \cfrac{2 \left \| \mathbf{H}^{\hat{f}}(t_j) \right\| \left \| \mathbf{H}(t_j) \right\|^{k} \left\| \mathbf{H}_f\right\|  (t_{j+1}-t_j)^k}{\left \| \mathbf{H}(t_j) \right\| }  \\
            & = \cfrac{ 2 \left\| \mathbf{H}_f\right\|\left\| \mathbf{H}^{\hat{f}}(t_j)\right\|}{\left \| \mathbf{H}(t_j) \right\|} \cdot \left(\cfrac{1}{1-\left \| \mathbf{H}(t_j) \right\| (t_{j+1} - t_j)} - 1 - \left \| \mathbf{H}(t_j) \right\| (t_{j+1} - t_j) \right) \\
            & = \cfrac{ 2 \left\| \mathbf{H}_f\right\|\left\| \mathbf{H}^{\hat{f}}(t_j)\right\|\left \| \mathbf{H}(t_j) \right\| (t_{j+1} - t_j)^2}{(1-\left \| \mathbf{H}(t_j) \right\| (t_{j+1} - t_j))}
        \end{aligned}
    \end{equation*}
    The first equality comes from $\frac{1}{1-x} = \sum_{j=0}^{\infty} x^{j}$ for $x\in [0,1]$.
\end{proof}

Finally, we aggregate all errors to the approximation ratio. If we hope that errors are less than $\epsilon$, then the upper bound of $t_{j+1}-t_j$ can be derived as follows:
\begin{equation}\label{eq: delta t in lyapunov 1 direct}
    \begin{array}{c}
          \begin{aligned}
              & \cfrac{ 2 \left\| \mathbf{H}_f\right\|\left\| \mathbf{H}^{\hat{f}}(t_j)\right\|\left \| \mathbf{H}(t_j) \right\| (t_{j+1} - t_j)^2}{1-\left \| \mathbf{H}(t_j) \right\| (t_{j+1} - t_j)} \\ 
              & + 2 \left\| \mathbf{H}_f \right \| \left( \sum_{k=1}^{n_1} \left\| \mathbf{H}^f_{k}(t_j)  \right\| + \sum_{k=1}^{n_2-1} \left\| \mathbf{H}_k^{\hat{f}}(t_j) \right\|  + \left\| \sum_{k=1}^{n_1} \mathbf{H}^f_k(t_j) + \sum_{k=1}^{n_2-1} \mathbf{H}_k^{\hat{f}} (t_j)\right\|  \right) (t_{j+1}-t_j) \le \epsilon 
          \end{aligned}
          \\
         t_{j+1}-t_j  \begin{cases}
             \le \cfrac{\epsilon}{B+C\epsilon+\sqrt{\epsilon(A-BC)}}& \text{if } A-BC \ge 0\\
              \le \cfrac{\epsilon}{B+C\epsilon+\sqrt{A \epsilon}} & \text{if }A-BC < 0
         \end{cases} \Rightarrow t_{j+1}-t_j \le \cfrac{\epsilon}{B+C\epsilon+\sqrt{A \epsilon}}
    \end{array}
\end{equation}
where
\begin{equation*}
    \begin{aligned}
        A &= 2 \left\| \mathbf{H}_f \right\|
               \left\| \mathbf{H}^{\hat{f}}(t_j) \right\|
               \left\| \mathbf{H}(t_j) \right\| \\[0.8em]
        B &= 2 \left\| \mathbf{H}_f \right\| 
             \bigg(
               \sum_{k=1}^{n_1} \left\| \mathbf{H}^f_{k}(t_j) \right\|
               + 
               \sum_{k=1}^{n_2-1} \left\| \mathbf{H}_k^{\hat{f}}(t_j) \right\|
               +
               \left\| 
                   \sum_{k=1}^{n_1} \mathbf{H}^f_{k}(t_j)
                   +
                   \sum_{k=1}^{n_2-1} \mathbf{H}_k^{\hat{f}}(t_j)
               \right\|
             \bigg) \\[0.8em]
        C &= \left\| \mathbf{H}(t_j) \right\|
    \end{aligned}
\end{equation*}


Totally, according to the above analyses, we design a quantum approximation algorithm for combinatorial optimization problems through Lyapunov control and simultaneously obtain the theoretical guarantees increase with the number of round increase. 
\begin{theorem} \label{thm::one-para-descrete}
    When we choose the Hamiltonian in Eq.(\ref{eq: Lyapunov 1 hamiltonian discrete}) and set the evolution time to Eq.(\ref{eq: delta t in lyapunov 1 direct}) in $j+1$-th round, then the approximation ratio increases after $j+1$-th round by Eq.(\ref{eq: approximation ratio lyapunov 1}). 
\end{theorem}

\section{Discrete Procedures for Lyapunov Function with Two Parameter} \label{Appendix::Two-Parameter-Discretion}
In this appendix, we present the discrete time formulation of the approximation ratio for the algorithm described by Eq. (\ref{eq::two_para_ratio}) under the two parameter Lyapunov framework. Firstly, we provide the formal definition of the discrete Lyapunov function.
\begin{definition}[Potential Function]
    The discrete Lyapunov function at time $t_j$ is as follows:
    \begin{equation}
        E(t_j) = x(t_j)\langle \psi(t_j)| \mathbf{H}_f | \psi(t_j) \rangle- y(t_j) \langle \psi^*| \mathbf{H}_f | \psi^* \rangle 
    \end{equation}
    The upper bound of the optimal solution at time $t_j$ is as follows:
    \begin{equation}
        0 \le \langle \psi^*| \mathbf{H}_f | \psi^* \rangle \le a(t_j) \langle \psi(t_j)| \mathbf{H}_f | \psi(t_j) \rangle+ b(t_j)\langle \psi(t_j)| \mathbf{Q}(t_j) | \psi(t_j) \rangle 
    \end{equation}
    where $x(t_j)$ and $y(t_j)$ are non-negative and non-decreasing sequences, $a(t_j)$ and $b(t_j)$ are non-negative sequences, and the upper bound sequence of the optimal solution can be constructed prior to the design of the algorithm. 
\end{definition}
Our target is to manipulate a 'non-decreasing' sequence $E(t_j)$ within a certain tolerant range that may disrupt monotonicity. The increment of the potential function from time $t_j$ to time $t_{j+1}$ is as follows:
\begin{equation}\label{eq::lyapunov increment 2}
    \begin{aligned}
        E(t_{j+1}) - E(t_j) & = x(t_{j+1})\langle \psi(t_{j+1})| \mathbf{H}_f | \psi(t_{j+1}) \rangle- y(t_{j+1}) \langle \psi^*| \mathbf{H}_f | \psi^* \rangle  - x(t_j)\langle \psi(t_j)| \mathbf{H}_f | \psi(t_j) \rangle + y(t_j) \langle \psi^*| \mathbf{H}_f | \psi^* \rangle \\
        & = x(t_{j+1}) \Big(\langle \psi(t_{j+1})| \mathbf{H}_f | \psi(t_{j+1}) \rangle - \langle \psi(t_j)| \mathbf{H}_f | \psi(t_j) \rangle\Big) \\
        & \quad+ \Big(x(t_{j+1})-x(t_j)\Big) \langle \psi(t_j)| \mathbf{H}_f | \psi(t_j) \rangle - \Big[ y(t_{j+1}) - y(t_j) \Big] \langle \psi^*| \mathbf{H}_f | \psi^* \rangle  \\
        & \ge x(t_{j+1}) \Big(\langle \psi(t_{j+1})| \mathbf{H}_f | \psi(t_{j+1}) \rangle - \langle \psi(t_j)| \mathbf{H}_f | \psi(t_j) \rangle\Big)\\
        & \quad + \Big(x(t_{j+1})-x(t_j) -  \big(y(t_{j+1}) - y(t_j)\big)a(t_j) \Big) \langle \psi(t_j)| \mathbf{H}_f | \psi(t_j) \rangle\\
        & \quad - \Big(y(t_{j+1}) - y(t_j)\Big) b(t_j)\langle \psi(t_j)| \mathbf{Q}(t_j) | \psi(t_j) \rangle  \\
        & \ge -Err(t_j)
    \end{aligned}
\end{equation}
If the total round is $N$, then the approximation ratio of the algorithm designing from the above potential function is as follows:
\begin{equation}\label{eq::approximation ratio c 2}
    \langle \psi(T)| \mathbf{H}_f | \psi(T) \rangle \ge \cfrac{y(T)-y(0)}{x(T)} \langle \psi^*| \mathbf{H}_f | \psi^* \rangle + \cfrac{x(0)}{x(T)}\langle \psi(0)| \mathbf{H}_f | \psi(0) \rangle -  \cfrac{\sum_{j=0}^{N-1} Err(t_j)}{x(T)} 
\end{equation}

Consequently, our target is to choose Hamiltonian $\mathbf{H}(t_j)$, $x(t_j)$, and $y(t_j)$ at each discrete time $t_j$ such that the difference $\frac{y(T)-y(0)}{x(T)}$ is as large as possible and the absolute value of errors are as small as possible. We start from Eq.(\ref{eq::lyapunov increment 2}) to derive Hamiltonian $\mathbf{H}(t_j)$, $x(t_j)$ and $y(t_j)$ at each time $t_j$. This implies that we are required to design three sequences $\mathbf{H}(t_j)$, $x(t_j)$, and $y(t_j)$. Generally speaking, this constitutes a large-scale optimization problem since the total round number $N$ could potentially be extremely large. Here, instead of directly tackling the large-scale optimization problem which is typically time-consuming, we impose certain restrictions on our algorithm to obtain a sufficiently good feasible solution. 

Similar with the single parameter Lyapunov function, we impose restrictions on our algorithm. Firstly, our algorithm evolves under a time-independent Hamiltonian during each time period, i.e., $\mathbf{H}(t) = \mathbf{H}(t_j)$ at each period $t \in [t_j,t_{j+1})$. Secondly, we impose the constraint that our algorithm can be efficiently implemented in quantum circuits via two components $e^{\mathbf{H}^{f}_1(t_j)},\cdots ,e^{\mathbf{H}^{f}_{n_1}(t_j)}$ and $e^{\mathbf{H}^{\hat{f}}_1(t_j)},\cdots ,e^{\mathbf{H}^{\hat{f}}_{n_2}(t_j)}$, where $\mathbf{H}^{f}_k(t_j)$ is commutative with $\mathbf{H}_f$, while $\mathbf{H}^{\hat{f}}_k(t_j)$ is not commutative with $\mathbf{H}_f$. For convenience, we denote $\mathbf{H}(t_j) = \mathbf{H}^f(t_j) + \mathbf{H}^{\hat{f}}(t_j) = \sum_{k=1}^{n_1} \mathbf{H}^{f}_k (t_j) + \sum_{k=1}^{n_2} \mathbf{H}^{\hat{f}}_k (t_j)$. Therefore, we have: 
\begin{equation}\label{eq: discrete marginal gain b b}
    \begin{aligned}
        & E(t_{j+1}) - E(t_j) = x(t_{j+1})\langle \psi(t_{j+1})| \mathbf{H}_f | \psi(t_{j+1}) \rangle - x(t_j)\langle \psi(t_j)| \mathbf{H}_f | \psi(t_j) \rangle  - \Big[ y(t_{j+1}) - y(t_j) \Big] \langle \psi^*| \mathbf{H}_f | \psi^* \rangle  \\
        & = x(t_{j+1})\langle \psi(t_j)|e^{i \mathbf{H}^f_1(t_j)(t_{j+1}-t_j)}\cdots e^{i \mathbf{H}^f_{n_1}(t_j)(t_{j+1}-t_j)} e^{i \mathbf{H}^{\hat{f}}_1(t_j)(t_{j+1}-t_j)}\cdots e^{i \mathbf{H}^{\hat{f}}_{n_2}(t_j)(t_{j+1}-t_j)} \\
        & \qquad \qquad \quad \mathbf{H}_f e^{-i \mathbf{H}^{\hat{f}}_{n_2}(t_j)(t_{j+1}-t_j)}\cdots e^{-i \mathbf{H}^{\hat{f}}_1 (t_j)(t_{j+1}-t_j)} e^{-i \mathbf{H}^f_{n_1}(t_j)(t_{j+1}-t_j)}\cdots e^{i \mathbf{H}^f_1(t_j)(t_{j+1}-t_j)}   | \psi(t_j)\rangle \\
        & \quad - x(t_j)\langle \psi(t_j)| \mathbf{H}_f | \psi(t_j) \rangle - \Big[ y(t_{j+1}) - y(t_j) \Big] \langle \psi^*| \mathbf{H}_f | \psi^* \rangle \\
        & \ge x(t_{j+1}) \Big[ \langle \psi(t_j)|e^{i \mathbf{H}(t_j)(t_{j+1}-t_j)} \mathbf{H}_f e^{-i \mathbf{H}(t_j)(t_{j+1}-t_j)} | \psi(t_j)\rangle - \langle \psi(t_j)| \mathbf{H}_f | \psi(t_j) \rangle \Big] \\
        & \quad -2 x(t_{j+1})\left\| \mathbf{H}_f \right \| \left( \sum_{k=1}^{n_1} \left\| \mathbf{H}^f_{k}(t_j) \right\| + \sum_{k=1}^{n_2-1} \left\| \mathbf{H}_k^{\hat{f}}(t_j)  \right\|  + \left\| \sum_{k=1}^{n_1} \mathbf{H}^f_k(t_j) + \sum_{k=1}^{n_2-1} \mathbf{H}_k^{\hat{f}} (t_j)\right\|  \right) (t_{j+1}-t_j) \\
        & \quad + \Big(x(t_{j+1})-x(t_j) -  \Big(y(t_{j+1}) - y(t_j)\Big)a(t_j) \Big) \langle \psi(t_j)| \mathbf{H}_f | \psi(t_j) \rangle\\
        & \quad - \Big(y(t_{j+1}) - y(t_j)\Big) b(t_j)\langle \psi(t_j)| \mathbf{Q}(t_j) | \psi(t_j) \rangle \\
    \end{aligned}
\end{equation}
Since the quantum state $|\psi(t_{j+1}) \rangle$ pretains to the subsequent discrete time, it remains unknown at the present time $t_j$. To address this, we employ the Taylor expansion to express $|\psi(t_{j+1}) \rangle$ in terms of $|\psi(t_j)\rangle,\mathbf{H}(t_j)$ and $(t_{j+1}-t_j)$. This approach also enables us to eliminate the objective function associated with the current quantum state, as it is independent of $\mathbf{H}(t_j)$. See details in Eq.(\ref{eq: taylor expansion}). Therefore, we can reformulate Eq.(\ref{eq: discrete marginal gain b b}) as follows:
\begin{equation*}
    \begin{aligned}
        & E(t_{j+1}) - E(t_j) \ge x(t_{j+1})\langle \psi(t_j)| i [ \mathbf{H}^{\hat{f}}(t_j), \mathbf{H}_f] | \psi(t_j) \rangle (t_{j+1}-t_j) - \Big(y(t_{j+1}) - y(t_j)\Big) b(t_j) \langle \psi(t_j)| \mathbf{Q}(t_j) | \psi(t_j) \rangle \\
        & \quad + \Big(x(t_{j+1})-x(t_j) -  (y(t_{j+1}) - y(t_j))a(t_j) \Big) \langle \psi(t_j)| \mathbf{H}_f | \psi(t_j) \rangle + x(t_{j+1})\sum_{k=2}^{\infty} \Delta^k(t_{j+1}-t_j)\\
        & \quad -2 x(t_{j+1})\left\| \mathbf{H}_f \right \| \left( \sum_{k=1}^{n_1} \left\| \mathbf{H}^f_{k}(t_j) \right\| + \sum_{k=1}^{n_2-1} \left\| \mathbf{H}_k^{\hat{f}}(t_j)  \right\|  + \left\| \sum_{k=1}^{n_1} \mathbf{H}^f_k(t_j) + \sum_{k=1}^{n_2-1} \mathbf{H}_k^{\hat{f}} (t_j)\right\|  \right) (t_{j+1}-t_j)
    \end{aligned}
\end{equation*}
where $\Delta^k (t_{j+1}-t_j)$ is denoted as $k$-th term in Taylor expansion. According to the above inequality, we derive a system to ensure the 'non-decreasing' sequence $E(t_j)$ with some tolerant ranges that may destroy monotonicity. Specially, the combination of the first two term is $0$, the third term is $0$, and the reminder terms can be bounded by $\epsilon$.
\begin{equation}
    \begin{cases}
        x(t_{j+1})\langle \psi(t_j)| i [ \mathbf{H}^{\hat{f}}(t_j), \mathbf{H}_f] | \psi(t_j) \rangle (t_{j+1}-t_j) - \Big(y(t_{j+1}) - y(t_j)\Big) b(t_j) \langle \psi(t_j)| \mathbf{Q}(t_j) | \psi(t_j) \rangle = 0\\
        x(t_{j+1})-x(t_j) -  \Big(y(t_{j+1}) - y(t_j)\Big)a(t_j) = 0
    \end{cases} 
\end{equation}
\begin{equation}
    \sum_{k=2}^{\infty} |\Delta^k(t_{j+1}-t_j)| + 2 \left\| \mathbf{H}_f \right \| \left( \sum_{k=1}^{n_1} \left\| \mathbf{H}^f_{k}(t_j) \right\| + \sum_{k=1}^{n_2-1} \left\| \mathbf{H}_k^{\hat{f}}(t_j)  \right\|  + \left\| \sum_{k=1}^{n_1} \mathbf{H}^f_k(t_j) + \sum_{k=1}^{n_2-1} \mathbf{H}_k^{\hat{f}} (t_j)\right\|  \right) (t_{j+1}-t_j) \le \frac{\epsilon}{x(t_{j+1})}
\end{equation}
Note that coming from the definition of the potential function $\Big(y(t_{j+1}) - y(t_j)\Big) b(t_j) \langle \psi(t_j)| \mathbf{Q}(t_j) | \psi(t_j) \rangle \ge 0$, it implies that we have to restrict $\langle \psi(t_j)| i [ \mathbf{H}^{\hat{f}}(t_j), \mathbf{H}_f] | \psi(t_j) \rangle (t_{j+1}-t_j) \ge 0$. Specially, in this work, we choose:
\begin{equation}\label{eq: Lyapunov 2 hamiltonian discrete}
    \mathbf{H}^{f}(t_j) = \sum_{k=1}^{n_1} \eta_k(t_j)\mathbf{H}^{f}_k, \quad \mathbf{H}^{\hat{f}}(t_j) = \sum_{k=1}^{n_2} \alpha_k(t_j) \mathbf{H}^{\hat{f}}_k
\end{equation}
where $\mathbf{H}^{f}_k$ and $\mathbf{H}^{\hat{f}}_k$ can be effectively implemented on quantum circuits. Using feedback-control and denoting $O_k(t_j) = \langle \psi(t_j)| i [\mathbf{H}^{\hat{f}}_k, \mathbf{H}_f] | \psi(t_j) \rangle$, we choose $\alpha_k(t_j) = \beta_k(t_j)O_k(t_j)$. If $\beta_k(t_j) \ge 0$, we have 
\begin{equation*}
    x(t_{j+1})\sum_{k=1}^{n_2} \beta_k(t_j)O^2_k(t_j) (t_{j+1}-t_j) = \Big(y(t_{j+1}) - y(t_j)\Big) b(t_j) \langle \psi(t_j)| \mathbf{Q}(t_j) | \psi(t_j) \rangle
\end{equation*}
Therefore, we have the following explicit formula of $x(t_j)$ and $y(t_j)$:
\begin{equation}\label{eq: x y in lyapunov 2 discrete}
    \begin{array}{c}
        x(t_{j+1}) = x(t_j) \cdot \cfrac{b(t_j) \langle \psi(t_j)| \mathbf{Q}(t_j) | \psi(t_j)}{b(t_j) \langle \psi(t_j)| \mathbf{Q}(t_j) | \psi(t_j) -a(t_j) \sum_{k=1}^{n_2} \beta_k(t_j)O^2_k(t_j) (t_{j+1}-t_j)} \\
        y(t_{j+1}) = y(t_j) + x(t_j) \cdot \cfrac{\sum_{k=1}^{n_2} \beta_k(t_j)O^2_k(t_j) (t_{j+1}-t_j)}{b(t_j) \langle \psi(t_j)| \mathbf{Q}(t_j) | \psi(t_j) -a(t_j) \sum_{k=1}^{n_2} \beta_k(t_j)O^2_k(t_j) (t_{j+1}-t_j)}
    \end{array}
\end{equation}

Next, we analyze errors. Firstly, considering the error from the $k$-th term in Taylor expansion about $t_{j+1}-t_j$, we have the following upper bound:
\begin{equation*}
    \sum_{k=2}^{\infty} \Delta^k(t_{j+1}-t_j) \le \sum_{k=2}^{\infty} |\Delta^k(t_{j+1}-t_j)| \le \cfrac{ 2 \left\| \mathbf{H}_f\right\|\left\| \mathbf{H}^{\hat{f}}(t_j)\right\|\left \| \mathbf{H}(t_j) \right\| (t_{j+1} - t_j)^2}{(1-\left \| \mathbf{H}(t_j) \right\| (t_{j+1} - t_j))}
\end{equation*}
Consequently, in each round, if our aim is to ensure that errors are bounded by $\epsilon$, it leads to an upper bound for the time period $t_{j+1}-t_j$ as dictated by the following inequality:
\begin{equation}
    \begin{aligned}
        & \cfrac{ 2 \left\| \mathbf{H}_f\right\|\left\| \mathbf{H}^{\hat{f}}(t_j)\right\|\left \| \mathbf{H}(t_j) \right\| (t_{j+1} - t_j)^2}{(1-\left \| \mathbf{H}(t_j) \right\| (t_{j+1} - t_j))}  \\ 
        & + 2 \left\| \mathbf{H}_f \right \| \left( \sum_{k=1}^{n_1} \left\| \mathbf{H}^f_{k}(t_j) \right\| + \sum_{k=1}^{n_2-1} \left\| \mathbf{H}_k^{\hat{f}}(t_j)  \right\|  + \left\| \sum_{k=1}^{n_1} \mathbf{H}^f_k(t_j) + \sum_{k=1}^{n_2-1} \mathbf{H}_k^{\hat{f}} (t_j)\right\|  \right) (t_{j+1}-t_j) \le \cfrac{\epsilon}{x(t_{j+1})}
    \end{aligned}
\end{equation}
\begin{equation}\label{eq: delta t in lyapunov 2 direct}
    t_{j+1}-t_j  \begin{cases}
             \le \cfrac{\frac{\epsilon}{x(t_{j+1})}}{B+C\frac{\epsilon}{x(t_{j+1})}+\sqrt{\frac{\epsilon}{x(t_{j+1})}(A-BC)}}& \text{if } A-BC \ge 0,\\
            \le \cfrac{\frac{\epsilon}{x(t_{j+1})}}{B+C\frac{\epsilon}{x(t_{j+1})}+\sqrt{\frac{A\epsilon}{x(t_{j+1})}}}  & \text{if }A-BC < 0,
         \end{cases} \Rightarrow t_{j+1}-t_j \le \cfrac{\epsilon}{Bx(t_{j+1})+C\epsilon+\sqrt{A x(t_{j+1})\epsilon}}
\end{equation}
where
\begin{equation*}
    \begin{aligned}
        A &= 2 \left\| \mathbf{H}_f \right\|
       \left\| \mathbf{H}^{\hat{f}}(t_j) \right\|
       \left\| \mathbf{H}(t_j) \right\| \\[0.8em]
        B &= 2 \left\| \mathbf{H}_f \right\|
             \bigg(
               \sum_{k=1}^{n_1} \left\| \mathbf{H}^f_{k}(t_j) \right\|
               +
               \sum_{k=1}^{n_2-1} \left\| \mathbf{H}_k^{\hat{f}}(t_j) \right\|
               +
               \left\|
                   \sum_{k=1}^{n_1} \mathbf{H}^f_k(t_j)
                   +
                   \sum_{k=1}^{n_2-1} \mathbf{H}_k^{\hat{f}}(t_j)
               \right\|
             \bigg) \\[0.8em]
        C &= \left\| \mathbf{H}(t_j) \right\|
    \end{aligned}
\end{equation*}

\begin{theorem} \label{thm::two-para-descrete}
    When we choose the Hamiltonian in Eq.(\ref{eq: Lyapunov 2 hamiltonian discrete}) and set the evolution time to Eq.(\ref{eq: delta t in lyapunov 2 direct}) in $j+1$-th round, then after $j+1$-th round, the parameter functions in the approximation ratio increases by Eq.(\ref{eq: x y in lyapunov 2 discrete}). 
\end{theorem}

\section{Proof for Constrained Combinatorial optimization}\label{appendix: discretization in lyapunov 1}
When considering constrained combinatorial optimization problems, there are some operators such that $|\psi (t) \rangle$ leaks outside the feasible domain. To ensure that the output of quantum algorithms remains feasible, a commonly adopted method is to initialize the algorithm in a quantum state $|\psi (0)\rangle$ on the Hilbert subspace for feasible solutions $\mathbb{C}_{\mathcal{C}}$, i.e., $|\psi (0)\rangle = \sum_{S \in \mathcal{C}} \alpha_{S} |\mathbf{s} \rangle$, where $|\mathbf{s}\rangle$ is the computational basis state encoding the feasible solution $S$. If unitary operators for algorithms are constrained so that they preserve the feasible subspace $\mathbb{C}_{\mathcal{C}}$, i.e., for $\forall |\psi\rangle \in \mathbb{C}_{\mathcal{C}}$,  $U |\psi\rangle \in \mathbb{C}_{\mathcal{C}}$, the output for the quantum algorithm corresponding to the unitary operator $U$ is guaranteed to be a feasible solution. This requirement naturally imposes specific constraints on the unitary operator $U$. Next, we rephrase one such constraint for feasibility-preserving operators, which has previously appeared in \citep{Hadfield2018}.
\begin{lemma}[A Rephrased Version of Proposition 5 in \citep{Hadfield2018}]\label{lemma:constraint}
    A unitary operator $e^{-i\int_0^t \mathbf{H}(s)ds}$ preserves eigenspaces of $\mathbf{H}_{\mathcal{C}}$ iff $e^{-i\int_0^t \mathbf{H}(s)ds}$ and $\mathbf{H}_{\mathcal{C}}$ are commutative, i.e., $[e^{-i\int_0^t \mathbf{H}(s)ds}, \mathbf{H}_{\mathcal{C}}] = 0$. Here, $\mathbf{H}_{\mathcal{C}}$ is the Hamiltonian related to constraints. For each computational basis state $|\mathbf{s}\rangle$ encoding the feasible solution $S$, we have $\mathbf{H}_{\mathcal{C}} |\mathbf{s}\rangle = \lambda_s |\mathbf{s}\rangle$ with $\lambda_s > 0$. For each computational basis state $|\mathbf{s}\rangle$ encoding the unfeasible solution $S$, we have $\mathbf{H}_{\mathcal{C}} |\mathbf{s}\rangle = \lambda_s |\mathbf{s}\rangle$ with $\lambda_s = 0$.
\end{lemma}

\begin{proof}
    For each quantum statement $|\psi\rangle$, we can decompose it into two component $|\psi \rangle = |\psi_{f}\rangle + |\psi_{uf}\rangle$ where $|\psi_{f}\rangle$ is in feasible solution space.
    \begin{itemize}
        \item '$\Rightarrow$': $e^{-i \int_0^t \mathbf{H}(s) ds} \mathbf{H}_{\mathcal{C}} |\psi \rangle=  e^{-i \int_0^t \mathbf{H}(s) ds} \mathbf{H}_{\mathcal{C}}(|\psi_{f}\rangle + |\psi_{uf}\rangle) = \lambda e^{-i \int_0^t \mathbf{H}(s) ds} |\psi_{f}\rangle = \mathbf{H}_{\mathcal{C}} e^{-i \int_0^t \mathbf{H}(s) ds} |\psi \rangle$. The last equality is because $e^{-i \int_0^t \mathbf{H}(s) ds}$ preserves eigenspaces. Therefore, we have $[e^{-i \int_0^t \mathbf{H}(s) ds}, \mathbf{H}_{\mathcal{C}}] = 0$. 
        \item '$\Leftarrow$': $0 = [e^{-i \int_0^t \mathbf{H}(s) ds}, \mathbf{H}_{\mathcal{C}}] |\psi_f \rangle = e^{-i \int_0^t \mathbf{H}(s) ds} \mathbf{H}_{\mathcal{C}} |\psi_f \rangle - \mathbf{H}_{\mathcal{C}}e^{-i \int_0^t \mathbf{H}(s) ds} |\psi_f \rangle = e^{-i \int_0^t \mathbf{H}(s) ds} \lambda |\psi_f \rangle - \mathbf{H}_{\mathcal{C}}e^{-i \int_0^t \mathbf{H}(s) ds} |\psi_f \rangle$. $0 = [e^{-i \int_0^t \mathbf{H}(s) ds}, \mathbf{H}_{\mathcal{C}}] |\psi_{uf} \rangle = e^{-i \int_0^t \mathbf{H}(s) ds} \mathbf{H}_{\mathcal{C}} |\psi_{uf} \rangle - \mathbf{H}_{\mathcal{C}}e^{-i \int_0^t \mathbf{H}(s) ds} |\psi_{uf} \rangle = - \mathbf{H}_{\mathcal{C}}e^{-i \int_0^t \mathbf{H}(s) ds} |\psi_{uf} \rangle$. That implies $\lambda e^{-i \int_0^t \mathbf{H}(s) ds}  |\psi_f \rangle = \mathbf{H}_{\mathcal{C}}e^{-i \int_0^t \mathbf{H}(s) ds} |\psi_f \rangle$ and $0 = e^{-i \int_0^t \mathbf{H}(s) ds}  |\psi_{uf} \rangle = \mathbf{H}_{\mathcal{C}}e^{-i \int_0^t \mathbf{H}(s) ds} |\psi_{uf} \rangle$, i.e., $e^{-i \int_0^t \mathbf{H}(s) ds}$ preserves eigenspaces.
    \end{itemize}
\end{proof}

\section{Proof of Lemma \ref{lm::ratio difference}}\label{Proof of Threom thm::ratio difference}
Lemma \ref{lm::ratio difference} introduces two new optimization problems within designing quantum algorithms that share the same optimal solution as the original problem. This result offers valuable insight into the construction of Lyapunov functions. The proof of this lemma is presented as follows:
\begin{proof}
   (1) Because $\mathbf{H}_1^*(t)$ is the optimal solution of the ratio maximization problem, then we have: 
    \begin{equation*}
            \langle \psi(0)| \mathcal{T}e^{i \int_0^t \mathbf{H}_1^*(s) ds }  \mathbf{H}_f \mathcal{T}e^{-i \int_0^t \mathbf{H}_1^*(s) ds }  |\psi(0)\rangle- \lambda^* \langle \psi^*| \mathbf{H}_f | \psi^* \rangle =0
    \end{equation*}
    If $\mathbf{H}_1^*(t)$ is not an optimal solution of the difference problem, i.e., there is a $\mathbf{\hat{H}}^*(t)$ such that 
    \begin{equation*}
            \langle \psi(0)| \mathcal{T}e^{i \int_0^t \mathbf{\hat{H}}^*(s) ds }  \mathbf{H}_f \mathcal{T}e^{-i \int_0^t \mathbf{\hat{H}}^*(s) ds }  |\psi(0)\rangle- \lambda^* \langle \psi^*| \mathbf{H}_f | \psi^* \rangle > 0\\
    \end{equation*}
    It is contradict with the optimal solution of the ratio problem.\\
    
    (2) Because $\mathbf{H}_2^*(t)$ is the optimal solution of the difference maximization problem, then we have:
    \begin{equation*}
            \cfrac{\langle \psi(0)| \mathcal{T}e^{i \int_0^t \mathbf{H}_2^*(s) ds }  \mathbf{H}_f \mathcal{T}e^{-i \int_0^t \mathbf{H}_2^*(s) ds }  |\psi(0)\rangle}{\langle \psi^*| \mathbf{H}_f | \psi^* \rangle +c^*}  = 1
    \end{equation*}
    If $\mathbf{H}_2^*(t)$ is not an optimal solution of the ratio problem, i.e., there is a $\mathbf{\hat{H}}^*(t)$ such that
    \begin{equation*}
            \cfrac{\langle \psi(0)| \mathcal{T}e^{i \int_0^t \mathbf{\hat{H}}^*(s) ds }  \mathbf{H}_f \mathcal{T}e^{-i \int_0^t \mathbf{\hat{H}}^*(s) ds }  |\psi(0)\rangle}{\langle \psi^*| \mathbf{H}_f | \psi^* \rangle +c^*} > 1\\
    \end{equation*}
    It is contradict with the optimal solution of the difference problem.\\
\end{proof}

\end{document}